\documentclass[11pt, a4paper, reqno]{amsart}

\usepackage{anysize}
\usepackage[utf8]{inputenc}
\usepackage{amsmath}
\usepackage{mathtools}
\usepackage{amssymb}
\usepackage{subcaption}
\usepackage{pdflscape}
\usepackage{hyperref}
\usepackage{thmtools, thm-restate}
\usepackage{xcolor}
\usepackage{mdframed}
\usepackage{mathrsfs}
\usepackage{enumerate}
\usepackage{complexity}					
\usepackage{afterpage}
\usepackage{bm}
\usepackage{quantikz}
\usepackage{comment}
\usepackage{tikz}
\usetikzlibrary{arrows,positioning,shapes.geometric}
\usepackage{algorithm}
\usepackage[noend]{algorithmic}

\marginsize{2.5cm}{2.5cm}{3cm}{3cm}

\newtheorem{theorem}{Theorem}
\numberwithin{theorem}{section}
\newtheorem{corollary}[theorem]{Corollary}
\newtheorem{proposition}[theorem]{Proposition}
\newtheorem{conjecture}[theorem]{Conjecture}
\newtheorem{claim}[theorem]{Claim}
\newtheorem{lemma}[theorem]{Lemma}
\newtheorem{observation}[theorem]{Observation}
\numberwithin{lemma}{section}

\theoremstyle{definition}
\newtheorem{definition}{Definition}
\numberwithin{definition}{section}

\AfterEndEnvironment{definition}{\noindent\ignorespaces}
\AfterEndEnvironment{theorem}{\noindent\ignorespaces}
\AfterEndEnvironment{conjecture}{\noindent\ignorespaces}
\AfterEndEnvironment{lemma}{\noindent\ignorespaces}
\AfterEndEnvironment{proof}{\noindent\ignorespaces}
\AfterEndEnvironment{corollary}{\noindent\ignorespaces}
\AfterEndEnvironment{proposition}{\noindent\ignorespaces}
\AfterEndEnvironment{observation}{\noindent\ignorespaces}
\AfterEndEnvironment{claim}{\noindent\ignorespaces}

\renewcommand{\a}{\mathfrak{a}}
\renewcommand{\b}{\mathfrak{b}}
\renewcommand{\c}{\mathfrak{c}}
\renewcommand{\d}{\mathfrak{d}}
\newcommand{\e}{\mathfrak{e}}
\newcommand{\f}{\mathfrak{f}}
\newcommand{\g}{\mathfrak{g}}
\newcommand{\h}{\mathfrak{h}}
\newcommand{\ii}{\mathfrak{i}}
\renewcommand{\j}{\mathfrak{j}}
\newcommand{\RR}{\mathbb{R}}			
\newcommand{\CC}{\mathbb{C}}			
\newcommand{\NN}{\mathbb{N}}			
\newcommand{\ZZ}{\mathbb{Z}}			
\newcommand{\Pauli}{\mathcal{P}}
\newcommand{\T}{\mathrm{T}}
\renewcommand{\AA}{\mathscr{A}}
\newcommand{\SU}{\mathrm{SU}}			
\renewcommand{\SL}{\mathrm{SL}}			
\newcommand{\GL}{\mathrm{GL}}			
\newcommand{\SO}{\mathrm{SO}}			
\newcommand{\U}{\mathrm{U}}			
\newcommand{\CNOT}{\mathrm{CNOT}}
\newcommand{\CZ}{\mathrm{CZ}}
\newcommand{\norm}[1]
	{\left\lVert#1\right\rVert}
\newcommand{\words}[1]{\langle #1 \rangle}	
\newcommand{\tr}{\mathrm{tr}}			
\newcommand{\opnorm}[1]%
	{\| #1 \|_{\mathrm{op}}}				

\newcommand{\odd}{\mathrm{odd}}
\newcommand{\even}{\mathrm{even}}

\newcommand{\IQP}{\mathrm{IQP}}
\newcommand{\CCC}{\mathrm{CCC}}

\newcommand{\TP}{\T_4 + \mathcal{P}}
\newcommand{\SampBPP}
	{\mathsf{SampBPP}}
\newcommand{\SampBQP}
	{\mathsf{SampBQP}}
\newcommand{\PostBPP}
	{\mathsf{PostBPP}}
\newcommand{\GadBQP}
	{\mathsf{GadBQP}}
\newcommand{\PostGadBQP}
	{\mathsf{PostGadBQP}}

\renewcommand{\S}{\mathcal{S}}          		
\newcommand{\anc}{\mathrm{anc}}			%
\newcommand{\post}{\mathrm{post}}			
\newcommand{\gad}{\mathrm{gad}}			

\renewcommand{\Pr}{\mathrm{Pr}}			
\renewcommand{\A}{\mathbf{A}}
\newcommand{\B}{\mathbf{B}}

\newcommand{\IsElementary}{\mathtt{IsElementary}}
\newcommand{\IsDiscrete}{\mathtt{IsDiscrete}}
\newcommand{\IsLoxodromic}{\mathtt{IsLoxodromic}}

\newcommand{\YES}{\mathtt{YES}}
\newcommand{\NO}{\mathtt{NO}}
\newcommand{\IDK}{\mathtt{IDK}}
\renewcommand{\C}{\mathbf{C}}			
\newcommand{\Q}{\mathbf{Q}}			
\renewcommand{\G}{\mathbf{G}}			
	
\newcommand{\defeq}{\coloneqq}			
\newcommand{\eref}[1]{Eq.~(\ref{#1})}		
\newcommand{\sref}[1]{Section~\ref{#1}}	
\newcommand{\appendref}[1]{Appendix~\ref{#1}}	
\newcommand{\thmref}[1]{Theorem~\ref{#1}}	
\newcommand{\algref}[1]{Algorithm~\ref{#1}}	
\newcommand{\propref}[1]%
{Proposition~\ref{#1}}					
\newcommand{\defref}[1]%
{Definition~\ref{#1}}					
\newcommand{\lemref}[1]{Lemma~\ref{#1}}	
\newcommand{\corref}[1]{Corollary~\ref{#1}}	
\newcommand{\conjref}[1]{Conjecture~\ref{#1}}	
\newcommand{\figref}[1]{Figure~\ref{#1}}		
\newcommand{\tabref}[1]{Table~\ref{#1}}		

\newcommand{\claimref}[1]{Claim~\ref{#1}}		

\title[A Criterion for Post-Selected Quantum Advantage]{A Criterion for Post-Selected Quantum Advantage}

\author[Chaitanya Karamchedu]{Chaitanya Karamchedu$^{\dagger}$}
\address{Department of Computer Science, University of Maryland}
\email{cdkaram@umd.edu}

\author[Matthew Fox]{Matthew Fox$^{\dagger}$}
\address{Department of Physics, University of Colorado Boulder}
\email{matthew.fox@colorado.edu}

\author[Daniel Gottesman]{Daniel Gottesman}
\address{Department of Computer Science, University of Maryland}
\email{dgottesm@umd.edu}

\thanks{$^\dagger$These authors contributed equally.}

\begin{document}

\begin{abstract}
Assuming the polynomial hierarchy is infinite, we prove a sufficient condition for determining if uniform and polynomial size quantum circuits over a non-universal gate set are not efficiently classically simulable in the weak multiplicative sense. Our criterion exploits the fact that subgroups of $\SL(2;\CC)$ are essentially either discrete or dense in $\SL(2;\CC)$. Using our criterion, we give a new proof that both instantaneous quantum polynomial (IQP) circuits and conjugated Clifford circuits (CCCs) afford a quantum advantage. We also prove that both commuting CCCs and CCCs over various fragments of the Clifford group afford a quantum advantage, which settles two questions of Bouland, Fitzsimons, and Koh. Our results imply that circuits over just $(U^\dagger \otimes U^\dagger) \CZ (U \otimes U)$ afford a quantum advantage for almost all $U \in \U(2)$.
\end{abstract}

\maketitle

\section{Introduction}

Quantum computers promise to outperform classical computers in certain computational tasks, such as simulating quantum systems \cite{Fey81} or solving abelian hidden subgroup problems like factoring integers \cite{CvD10, Sho99}. In computational complexity, this promise is largely encapsulated in the difficult conjecture that $\BPP \neq \BQP$ \cite{AB09, NC11}. However, experimentally implementing most quantum algorithms, especially those that decide suspected $\BQP \backslash \BPP$ languages, is very difficult. For example, the largest number ever factored using Shor's algorithm is $21$, and this occurred over a decade ago \cite{MLLL+12}.

Given both the theoretical challenges in formally proving a quantum advantage, as well as the experimental challenges in practically realizing fault-tolerant universal quantum computation, there is considerable interest in ``restricted models'' of quantum computation. These are models of quantum computation that are non-universal by design, but which are nevertheless able to perform computational tasks that no efficient classical computer can, at least under standard complexity assumptions.

Notable examples of restricted models include: non-adaptive linear optics \cite{AA11}, constant-depth quantum circuits \cite{TD04, BGK18}, instantaneous quantum polynomial (IQP) circuits \cite{BJS11}, conjugated Clifford circuits (CCCs) \cite{BFK18}, and quantum circuits with one clean qubit \cite{KL98, FKM+18}. Interestingly, under the standard complexity assumption that the polynomial hierarchy is infinite \cite{AB09}, all of these restricted models can perform \emph{sampling} tasks that no efficient classical computer can, and this is despite the fact, but in no way in contradiction to it, that some of these restricted models can probably not solve \emph{decision} tasks outside of $\BPP$.

From a theoretical perspective, restricted models are interesting because they seem to straddle the purported boundary between $\BPP$ and $\BQP$, analogous to how $\NP$-intermediate problems straddle the purported boundary between $\P$ and $\NP$ \cite{Lad75}. A deeper understanding of the computational complexity of restricted models, therefore, might inform the relationship between $\BPP$ and $\BQP$. 

From an experimental perspective, restricted models are particularly useful because they make the task of demonstrating a quantum advantage distinct from realizing a universal quantum computer. This follows because many restricted models that can perform ``hard'' sampling tasks are substantially easier to implement fault-tolerantly than a full-blown universal quantum computer \cite{AA11}.

There are many ways to formulate a restricted model, but arguably the most straightforward is to consider quantum circuits over a \emph{non}-universal gate set $\S$, such as
$$
\S_{\IQP} \defeq \big\{H T H, (H \otimes H) \CZ (H \otimes H)\big\},
$$ 
which underlies a subset of IQP circuits, and, for any single-qubit unitary $U$,
$$
\S_{\CCC}(U) \defeq \big\{U^\dagger H U, U^\dagger S U, (U^\dagger \otimes U^\dagger) \CNOT (U \otimes U)\big\},
$$ 
which underlies all $U$-CCCs. 

Non-universal gate sets are interesting for a number of reasons. On one hand, they are ``rare'' in the measure-theoretic sense that almost every set of qudit gates is universal \cite{Llo95}. On the other hand, non-universal gate sets are somewhat easy to contrive (particularly in the context of $U$-CCCs where the conjugating $U$ can be any single-qubit unitary) and while it is possible to decide if a gate set is universal \cite{Iva06}, doing so is generally difficult \cite{SK17}. What is especially interesting, though, is that for many non-universal gate sets $\S$, there is a \emph{post-selected} circuit over $\S$ that can decide a $\PP$-complete language. By a standard complexity argument that is detailed in our \propref{prop:hardness}, this implies that the original, \emph{non}-post-selected circuit over $\S$ can perform a sampling task that no efficient classical computer can, unless the polynomial hierarchy collapses. Given this, we ask: under standard complexity assumptions, is there a simple ``catch all'' (or at the very least a ``catch many'') criterion for understanding when post-selected circuits over a non-universal gate set can decide a $\PP$-complete language?

In this paper, we answer this question in the affirmative by exhibiting an elementary algorithm that does just this. Our techniques exploit the fact that the subgroups of $\SL(2;\CC)$ are stringently constrained. In particular, we show that for any non-universal gate set $\S$, if a finite collection $\Gamma$ of invertible, single-qubit, post-selection gadgets over $\S$ generates a group $\words{\Gamma}$, then $\words{\Gamma}$ is essentially either discrete or dense in a closed subgroup of $\SL(2;\CC)$. Thus, roughly speaking, if $\words{\Gamma}$ is \emph{non}-discrete, then when augmented with post-selection, circuits over $\S$ can simulate a universal gate set and therefore simulate any \emph{universal} quantum computer. Consequently, to determine if post-selected circuits over $\S$ can decide a $\PP$-complete language, it more or less suffices to determine if $\words{\Gamma}$ is discrete or not. Modulo several group-theoretic details, this is all that our algorithm does.

Altogether, we establish the following straightforward procedure for determining if uniform and polynomial size quantum circuits over a non-universal gate set $\S$ can perform a computational task that no efficient classical computer can:

\begin{enumerate}
    \item Using circuits over $\S$, build a set of invertible, single-qubit, post-selection gadgets $\Gamma$. 
    \item Check if $\Gamma$ is closed under inverses. If not, add more gadgets until it is.
    \item Using our criterion, perform a set of simple computations to check if $\words{\Gamma}$ is a dense subgroup of $\SL(2; \CC)$. If yes, then efficient classical computers cannot simulate uniform and polynomial size quantum circuits over $\S$ unless the polynomial hierarchy collapses.
\end{enumerate}

Unfortunately, there is no obvious systematic way to identify ``good'' gadgets in step 1 (in the sense that they will return ``yes'' in step 3), nor is there an obvious systematic way to build inverse gadgets in step 2. Nevertheless, if both ``good'' and inverse gadgets exist, then our experience suggests a na\"ive computer search is usually sufficient to identify them, especially for small gadget sizes.

Using this procedure, we answer two questions raised by Bouland et al. in \cite{BFK18} by proving that efficient classical computers can neither simulate \emph{commuting} conjugated Clifford circuits, which are in some sense the ``intersection'' of CCCs and IQP circuits, nor can they simulate CCCs when the interstitial Clifford circuit is restricted to one of several ``fragments'' of the Clifford group (as classified in \cite{GS22}), such as Clifford circuits made entirely of $Z$ and $\CZ$ gates, or even just $\CZ$ gates. We remark that the topological idea of finding sufficient conditions for $\Gamma$ to generate a dense subgroup of $\SL(2;\CC)$ (as opposed to, say, $\SU(2)$) was inspired by the techniques in \cite{AAEL08} and \cite{BMZ16}.

This paper is structured as follows. In \sref{sec:prelims}, we review several standard notions such as efficient classical and quantum computers, as well as their post-selected counterparts. There, we also introduce most of our notation and central definitions. In \sref{sec:mainresults}, we state and prove our criterion for post-selected quantum advantage. In \sref{sec:grouptheory}, we introduce several results to do with the group theory of $\SL(2;\CC)$, which together afford a practical algorithm that implements our criterion. In \sref{sec:applications}, we demonstrate the utility of our criterion by proving that IQP circuits, CCCs, commuting CCCs, and CCCs over several fragments of the Clifford group can all perform a sampling task that no efficient classical computer can, unless the polynomial hierarchy collapses. Finally, in \sref{sec:discussion}, we discuss several open questions, some of which serve as a means of improving our criterion.

\section{Preliminaries}
\label{sec:prelims}

In this paper, $\{0,1\}^n$ is the set of all (binary) strings with length $n$, $\{0,1\}^*$ is the set of all strings with finite length, $x_i$ is the $i$th bit of the string $x$, $x.y$ is the concatenation of the strings $x$ and $y$, $x_{[k\ :\ \ell]}$ is the substring $x_{k}\dots x_\ell$ of the string $x$, $[k]$ is the set $\{1,2,\dots, k\}$, $\NN$ is the set of positive integers, $\ZZ$ is the set of integers, $\RR$ is the set of real numbers, $(a,b)$ is the open interval $\{x \in \RR \mid a < x < b\}$, $[a,b]$ is the closed interval $\{x \in \RR \mid a \leq x \leq b\}$, $\CC$ is the set of complex numbers, $\words{\Gamma}$ is the group generated by a set $\Gamma$ (where it is implicitly assumed that $\Gamma$ is closed under inverses), $I_k$ is the $2^k \times 2^k$ identity matrix, and $\Sigma_k\P$ is the $k$th level of the polynomial hierarchy, $\PH$. All other notation is defined along the way, except for standard complexity classes like $\P$ and $\PP$ and standard Lie groups like $\SO(n)$ and $\SL(n; \CC)$. Finally, every topological statement like ``such-and-such is dense in $\SL(2;\CC)$'' is meant with respect to the operator norm topology on $\SL(2;\CC)$.

\subsection{Notions of Simulation}

Let $\A$ and $\B$ be classical or quantum algorithms that, for all $n \in \NN$, output $y \in \{0,1\}^{f(n)}$ on input $x \in \{0,1\}^n$ with probabilities $\Pr[\A(x) = y]$ and $\Pr[\B(x) = y]$, respectively. Here, $f : \NN \rightarrow \NN$ characterizes the length of the strings that $\A$ and $\B$ output on a given input size.

We begin by defining a particular type of simulation, known as a \emph{weak multiplicative simulation}, where the word ``weak'' refers to the fact that the simulation of $\B$ by $\A$ is only required to simulate the output distribution of $\B$ on a given input, as opposed to the much stronger notion where $\A$ is required to output the probability that $\B$ outputs a particular string on a given input.

\begin{definition}
We say \emph{$\A$ simulates $\B$ to within multiplicative error $\epsilon \geq 0$} if and only if (iff) for all $n \in \NN$, all $x \in \{0,1\}^n$, and all $y \in \{0,1\}^{f(n)}$, 
$$
\frac{1}{1 + \epsilon} \Pr[\B(x) = y] \leq \Pr[\A(x) = y] \leq (1 + \epsilon) \Pr[\B(x) = y].
$$
\end{definition}

In other words, $\A$ simulates $\B$ to within multiplicative error $\epsilon$ iff on every valuation (and hence in the worst case), $\A$ samples the output distribution of $\B$ to within $\epsilon$. While this is the notion of simulation we employ in this paper, we admit that it is rather restrictive. For example, if $\B$ has two likely events with probability $1/2 - 1/2^n$ and two unlikely events with probability $1/2^n$, insisting that $\A$ simulates $\B$ to within small multiplicative error means $\A$ must approximate the two unlikely events essentially just as well as it approximates the two likely events. However, it should be difficult to actually differentiate between the correct distribution underlying $\B$ and the distribution where the two likely events have probability $1/2$ each.

Incidentally, from an experimental perspective, a better notion of simulation is that of a \emph{weak additive simulation}. In this case, it is merely required that \emph{on the average} (as opposed to on every valuation) $\A$ samples the output distribution of $\B$ to within some error $\epsilon$. This notion is better experimentally because the threshold theorem for quantum fault tolerance only guarantees the robustness of quantum computation up to \emph{additive} error between the target and actualized output distributions \cite{AB97}.

Evidently, if $\A$ simulates $\B$ to within multiplicative error $\epsilon$, then $\A$ simulates $\B$ to within additive error $\epsilon$ as well. However, the converse is not generally true because $\A$ can sample the output distribution of $\B$ to within error $\epsilon$ \emph{on the average} but not on every valuation \cite{NLD+22}. Nevertheless, if $\epsilon = 0$, then $\Pr[\A(x) = y] = \Pr[\B(x) = y]$ for all $n \in \NN$, all $x \in \{0,1\}^n$, and all $y \in \{0,1\}^{f(n)}$. In this case, the weak additive and multiplicative senses entail the same thing, so we say that \emph{$\A$ exactly simulates $\B$}

We now discuss a variety of algorithmic models, starting with efficient classical computers.

\subsection{Classical Computers}

On account of the extended Church-Turing thesis \cite{AB09}, an \emph{efficient classical computer} $\C$ is a probabilistic and polynomial time Turing machine. For every $\C$, there is a polynomial $f_\C : \NN \rightarrow \NN$ that specifies the \emph{output size of $\C$}, i.e., the length of the output strings of $\C$ on a given input size. For all $n \in \NN$, the probability $\C$ outputs $y \in \{0,1\}^{f_\C(n)}$ on input $x \in \{0,1\}^n$ is $\Pr[\C(x) = y]$, where the probability is over the internal randomness of $\C$.

\subsection{Post-Selected Classical Computers}

We now define the notion of an ``efficient post-selected classical computer''.

\begin{definition}
An \emph{efficient post-selected classical computer} $\mathbf{C}_\post$ is a tuple $(\C, \post)$, where $\C$ is an efficient classical computer and $\post : \NN \rightarrow \{0,1\}^*$ is a total and polynomial time computable function that specifies a post-selection string on a given input size. For all $n \in \NN$ and all $x \in \{0,1\}^n$, $\C$ and $\post$ satisfy $f_{\C_\post}(n) \defeq f_\C(n) - |\post(n)| > 0$ and
$$
\Pr\big[\C(x)_{[f_{\C_\post}(n) + 1\ :\ f_{\C}(n)]} = \post(n)\big] = \sum_{z \in \{0,1\}^{f_{\C_\post}(n)}} \Pr\left[\C(x) = z.\post(n)\right] \neq 0.
$$
The former condition ensures that the \emph{output size of $\C_\post$, $f_{\C_\post}$,} is always positive, and the latter condition ensures that the probability that $\C$ outputs a string whose last $|\post(n)|$ bits equal $\post(n)$ is nonzero. This way, for all $n \in \NN$, the probability $\C_{\post}$ outputs $y \in \{0,1\}^{f_{\C_\post}(n)}$ on input $x \in \{0,1\}^n$ is well-defined:
\begin{align*}
\Pr\left[\C_\post(x) = y \right] &\defeq \Pr\big[\C(x)_{[1\ :\ f_{\C_\post}(n)]} = y \mid \C(x)_{[f_{\C_\post}(n) + 1\ :\ f_{\C}(n)]} = \post(n)\big]\\
&= \frac{\Pr\left[\C(x) = y.\post(n)\right]}{\Pr\big[\C(x)_{[f_{\C_\post}(n) + 1\ :\ f_\C(n)]} = \post(n)\big]}.
\end{align*}
Notice, this is the conditional probability that the first $f_{\C_\post}(n)$ bits of $\C(x)$ equal $y$, given that the last $|\post(n)|$ bits of $\C(x)$ equal $\post(n)$.
\end{definition}

The set of languages that efficient post-selected classical computers can decide is often called $\PostBPP$ \cite{BJS11}, but this class is also known as $\BPP_{\text{path}}$ \cite{HHT97}. 
\begin{definition}
The class $\PostBPP$ consists of all languages $L$ for which there is an efficient post-selected classical computer $\C_\post$ and $\delta \in (0,1/2)$ such that for all $x \in \{0,1\}^*$, $\Pr[\C_\post(x)_1 = L(x)] \geq 1/2 + \delta,$ where $L(\cdot)$ is the indicator function of $L$.
\end{definition}

Crucially, $\PostBPP$ is actually independent of $\delta$ because efficient post-selected classical computers can compute the majority function, and hence they can amplify the acceptance probabilities \cite{BJS11}. This implies that if $L \in \PostBPP$, then \emph{for all} $\delta \in (0, 1/2)$, there is an efficient post-selected classical computer $\C_\post$ that decides $L$ to within bounded error $1/2 + \delta$.

\subsection{Quantum Computers}

We now recall several notions to do with the quantum circuit model.

\begin{definition}
\label{def:gateset}
A \emph{gate set} $\S$ is a finite subset of $\bigcup_{\ell \in [k]}\U(2^\ell)$ for some fixed $k > 1$ that contains at least one entangling gate.\footnote{If $\S$ does not contain an entangling gate, then uniform and polynomial size circuits over $\S$ with pure state inputs are efficiently classically simulable, so there is no hope for a quantum advantage \cite{JL03}.} We say $\S$ is \emph{universal} iff there exists $\ell \in [k]$ for which $\words{\S \cap \U(2^\ell)}$ is dense in $\U(2^\ell)$.\footnote{Interestingly, a universal gate set $\S$ can approximate all gates \emph{efficiently} as a function of accuracy, even if $\S$ is not closed under inverses, due to the \emph{inverse-free} Solovay-Kitaev theorem \cite{BG21}.} Otherwise, $\S$ is \emph{non-universal}.
\end{definition}

Note, whenever we say ``gate set'' in this paper, we implicitly mean a ``not necessarily universal gate set''.

Given a gate set, one can build circuits over it.

\begin{definition}
An \emph{$n$-qubit quantum circuit $Q_n$ over a gate set $\S$} is an operator in $\U(2^n)$ that admits the product decomposition
$$
Q_n = U_{d_n} \dots U_1.
$$
Here, each $U_1, \dots, U_{d_n} \in \U(2^n)$ is a tensor product of operators in $\S \cup \{I_1\}$. The total number of operators in $\S$ that make up $Q_n$ is the \emph{size of $Q_n$} and $d_n$ is the \emph{depth of $Q_n$}.
\end{definition}

We now define what we mean by an ``efficient quantum computer''.

\begin{definition}
An \emph{efficient quantum computer $\Q$ over a gate set $\S$} is a triple $(Q, \S, \anc)$, where $\anc : \NN \rightarrow \{0,1\}^*$ is a total and polynomial time computable function whose purpose is to specify an ancilla string on a given input size and $Q = (Q_n)_{n \in \NN}$ is a uniform family of polynomial size, $(n + |\anc(n)|)$-qubit quantum circuits $Q_n$ over $\S$. For all $n \in \NN$, the \emph{output size of $\Q$} is the polynomial $f_\Q(n) \defeq n + |\anc(n)|$ and the probability that $\Q$ outputs $y \in \{0,1\}^{f_\Q(n)}$ on input $x \in \{0,1\}^n$ is $\Pr\left[\Q(x) = y\right] \defeq \left|\bra{y}Q_{n}\ket{x} \otimes \ket{\anc(n)}\right|^2.$
\end{definition}

\subsection{Post-Selected Quantum Computers}

We now introduce the notion of an ``efficient post-selected quantum computer''.

\begin{definition}
An \emph{efficient post-selected quantum computer $\Q_\post$ over a gate set $\S$} is a 4-tuple $(Q,\S,\anc, \post)$, where $\Q = (Q,\S,\anc)$ is an efficient quantum computer and $\post : \NN \rightarrow \{0,1\}^*$ is a total and polynomial time computable function whose purpose is to specify a post-selection string on a given input size. For all $n \in \NN$ and all $x \in \{0,1\}^n$, $\Q$ and $\post$ satisfy $f_{\Q_\post}(n) \defeq f_\Q(n) - |\post(n)| > 0$ and
$$
\Pr\big[\Q(x)_{[f_{\Q_\post}(n) + 1\ :\ f_\Q(n)]} =  \post(n)\big] = \sum_{z \in \{0,1\}^{f_{\Q_\post}(n)}} \Pr[\Q(x) = z.\post(n)] \neq 0.
$$
The former condition ensures that the \emph{output size of $\Q_\post$, $f_{\Q_\post}$,} is always positive, and the latter condition ensures that the probability that $\Q$ outputs a string whose last $|\post(n)|$ bits equal $\post(n)$ is nonzero. This way, the probability the last $|\post(n)|$ registers of $\Q(x)$ have $\ket{\post(n)}$ as the state is non-zero, which is necessary for the probability that $\Q_{\post}$ outputs $y \in \{0,1\}^{f_{\Q_\post}(n)}$ on input $x \in \{0,1\}^n$ to be well-defined:
\begin{align*}
\Pr\left[\Q_\post(x) = y \right] &\defeq \Pr\big[\Q(x)_{[1\ :\ f_{\Q_\post}(n)]} = y \mid \Q(x)_{[f_{\Q_\post}(n) + 1\ :\ f_{\Q}(n)]} = \post(n)\big]\\
&= \frac{\Pr\left[\Q(x) = y.\post(n)\right]}{\Pr\big[\Q(x)_{[f_{\Q_\post}(n) + 1\ :\ f_{\Q}(n)]} = \post(n)\big]}.
\end{align*}
Notice, this is the conditional probability that the upper $f_{\Q_\post}(n)$ registers of $\Q(x)$ have $\ket{y}$ as the state, given that the lower $|\post(n)|$ registers of $\Q(x)$ have $\ket{\post(n)}$ as the state.
\end{definition}

The set of languages that efficient post-selected quantum computers over a universal gate set $\S$ can decide is usually called $\PostBQP$ \cite{Aar05}. We extend this definition as follows, which allows for non-universal $\S$.
\begin{definition}
Let $\S$ be a gate set. The class $\PostBQP(\S)$ consists of all languages $L$ for which there is an efficient post-selected quantum computer $\Q_\post$ over $\S$ and $\delta \in (0,1/2)$ such that for all $x \in \{0,1\}^*$, $\Pr[\Q_\post(x)_1 = L(x)] \geq 1/2 + \delta.$
\end{definition}

As already remarked, if $\S$ is universal, then $\PostBQP(\S) = \PostBQP$. In this case, $\PostBQP$ is independent of $\delta$ because efficient post-selected quantum computers over a universal gate set can compute the majority function \cite{Aar05}. Therefore, if $L \in \PostBQP$, then \emph{for all} $\delta \in (0,1/2)$, there is an efficient post-selected quantum computer $\Q_\post$ that decides $L$ to within bounded error $1/2 + \delta$. However, for non-universal $\S$, $\PostBQP(\S)$ is not necessarily independent of $\delta$ in this sense, because $\S$ may be sufficiently constrained that no circuit over $\S$ can approximate the majority function in any useful sense. Therefore, a priori, if $L \in \PostBQP(\S)$ for some non-universal $\S$, then there merely \emph{exists} $\delta \in (0,1/2)$ and an efficient post-selected quantum computer $\Q_\post$ that decides $L$ to within bounded error $1/2 + \delta$. 

\subsection{Postselection Gadgets and Gadget Quantum Computers}

We now define ``post-selection gadgets'', which give rise to particular post-selected quantum computers.
\begin{definition}
Let $\S$ be a gate set. A \emph{$j$-to-$k$ post-selection gadget over $\S$} (or, if $\S$ is contextually clear, a \emph{$j$-to-$k$ post-selection gadget}, or just a \emph{$j$-to-$k$ gadget}) is a map $\g_{j, k} : \CC^{2^j} \rightarrow \CC^{2^j}$ that acts on a $k$-qubit system as follows:
\begin{enumerate}[(i)]
\item Introduce a set $A$ of $j - k$ ancillae in the state $\ket{a_1 \dots a_{j - k}}_A$, where each $a_i \in \{0,1\}$.
\item Apply a $j$-qubit circuit $Q(\g_{j,k}) : \CC^{2^{j}} \rightarrow \CC^{2^{j}}$ over $\S$ to both the system and ancillae.
\item Post-select on a set $B$ of $j - k$ qubits being in the state $\ket{b_1 \dots b_{j - k}}_B$, where each $b_i \in \{0,1\}$.
\end{enumerate}
For example, the following circuit fragment is a $4$-to-$1$ post-selection gadget $\g_{4,1}$ over $\S$:
\newline
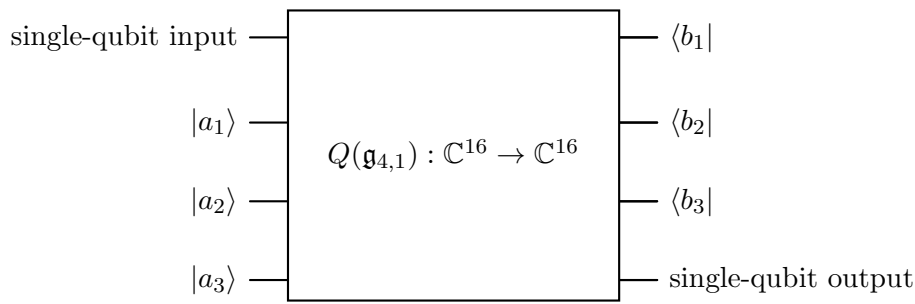
\begin{center}
\begin{quantikz}
\lstick{single-qubit input} & \gate[4]{\quad \text{$Q(\g_{4,1}) : \CC^{16} \rightarrow \CC^{16}$} \quad} & \qw \rstick{\bra{b_1}} \\
\lstick{\ket{a_1}} & \ghost{S} & \qw \rstick{\bra{b_2}} \\
\lstick{\ket{a_2}} & \ghost{S} & \qw \rstick{\bra{b_3}}\\
\lstick{\ket{a_3}} & \ghost{S} & \qw \rstick{single-qubit output}
\end{quantikz}
\end{center}
\end{definition}

Given a $j$-to-$k$ post-selection gadget $\g_{j, k}$, we define its \emph{action} $\AA(\g_{j,k})$ as the $2^k \times 2^k$ matrix 
$$
\AA(\g_{j, k}) \defeq \prescript{}{B}{\bra{b_1 \dots b_{j - k}}}\, Q(\g_{j,k}) \ket{a_1 \dots a_{j - k}}_A.
$$
If $\det \AA(\g_{j,k}) \neq 0$, then the \emph{normalized action} of $\g_{j,k}$ is the unit-determinant $2^k \times 2^k$ matrix
$$
\tilde{\AA}(\g_{j,k}) \defeq \frac{\AA(\g_{j,k})}{(\det \AA(\g_{j,k}))^{2^{-k}}}.
$$
Finally, the set $\gad_k(\S)$ consists of the normalized actions of all $j$-to-$k$ post-selection gadgets over $\S$ with non-zero determinant, i.e.,
$$
\gad_k(\S) \defeq \bigcup_{j \in \NN} \big\{\tilde{\AA}(\g_{j,k}) \mid \text{$\g_{j,k}$ is a $j$-to-$k$ post-selection gadget with $\det \AA(\g_{j,k}) \neq 0$}\big\}.
$$

\begin{observation}
For every $k \in \NN$ and every gate set $\S$, $I_k \in \gad_k(\S)$ and $\gad_k(\S) \subset \SL(2^k; \CC)$. Therefore, if $\Gamma \subset \gad_k(\S)$ and $\words{\Gamma}$ is closed under inverses, then $\words{\Gamma}$ is a subgroup of $\SL(2^k;\CC)$.
\end{observation}

In practice, proving that $\words{\Gamma}$ is closed under inverses is difficult. This is because while it is easy to contrive gadgets, it is not always easy to contrive their ``inverse''.
\begin{definition}
\label{def:inversegadget}
Let $\g = \g_{j,k}$ be a $j$-to-$k$ gadget over a gate set $\S$ such that $\det \AA(\g) \neq 0$. We call a $j'$-to-$k$ gadget $\g^{-1} = \g_{j',k}$ over $\S$ an \emph{inverse gadget of $\g$} iff $\AA(\g)^{-1} \in \words{\AA(\g), \AA(\g^{-1})}$.
\end{definition}

A means of improving our results is to show that for any $j$-to-$k$ post-selection gadget $\g = \g_{j,k}$ over a gate set $\S$ such that $\det \AA(\g) \neq 0$, there exists an inverse gadget of $\g$ over $\S$. If this is true, then for any finite set $\Gamma \subset \gad_k(\S)$, there exists another finite set $\Gamma' \subset \gad_k(\S)$ such that $\Gamma \subseteq \Gamma'$ and $\words{\Gamma'}$ is a subgroup of $\SL(2;\CC)$. We leave this as an open question (see \conjref{conj:gadgetconj}).

We now define ``gadget quantum circuits'', which are essentially quantum circuits defined over $\S \cup \Gamma$ for some gate set $\S$ and some finite collection of gadgets $\Gamma \subset \bigcup_{k \in \NN}\gad_k(\S)$.

\begin{definition}
Let $\Gamma \subset \bigcup_{\ell \in [k]} \SL(2^\ell;\CC)$ be a finite subset for some fixed $k > 1$. An \emph{$n$-qubit gadget quantum circuit $Q_n$ over $\Gamma$} is an operator in $\SL(2^n;\CC)$ that admits the product decomposition
$$
Q_n = \omega_{d_n} \dots \omega_1,
$$
where each $\omega_1, \dots, \omega_{d_n} \in \SL(2^n;\CC)$ is a tensor product of operators in $\Gamma \cup \{I_1\}$. The total number of operators in $\Gamma$ that make up $Q_n$ is the \emph{size of $Q_n$} and $d_n$ is the \emph{depth of $Q_n$}. 
\end{definition}

Gadget quantum circuits naturally define a ``gadget quantum computer''.
\begin{definition}
An \emph{efficient gadget quantum computer $\G$ over a gate set $\S$} is a 4-tuple $(G, \S, \Gamma, \anc)$, where $\Gamma \subset \bigcup_{k \in \NN} \gad_k(\S)$ is a finite set, $\anc : \NN \rightarrow \{0,1\}^*$ is a total and polynomial time computable function, and $Q = (Q_n)_{n \in \NN}$ is a uniform family of polynomial size $(n + |\anc(n)|)$-qubit gadget quantum circuits $Q_n$ over $\S \cup \Gamma$. For all $n \in \NN$, the \emph{output size of $\Q_\gad$} is the polynomial $f_{\Q_\gad}(n) \defeq n + |\anc(n)|$ and the probability that $\Q_\gad$ outputs $y \in \{0,1\}^{f_{\Q_\gad}(n)}$ on input $x \in \{0,1\}^n$ is $\Pr[\Q_\gad(x) = y] \defeq \left|\bra{y}Q_{n}\ket{x} \otimes \ket{\anc(n)}\right|^2.$
\end{definition}

\subsection{Post-Selected Gadget Quantum Computers}

Finally, we introduce the notion of a ``post-selected gadget quantum computer'', which is to a gadget quantum computer what a post-selected quantum computer is to a quantum computer.

\begin{definition}
An \emph{efficient post-selected gadget quantum computer $\G_{\post}$ over a gate set $\S$} is a 5-tuple $(G,\S,\Gamma, \anc, \post)$, where $\G = (G,\S, \Gamma,\anc)$ is an efficient gadget quantum computer and $\post : \NN \rightarrow \{0,1\}^*$ is a total and polynomial time computable function whose purpose is to specify a post-selection string on a given input size. For all $n \in \NN$ and all $x \in \{0,1\}^n$, $\G$ and $\post$ satisfy $f_{\G_{\post}}(n) \defeq f_{\G}(n) - |\post(n)| > 0$ and
$$
\Pr\big[{\G(x)}_{[f_{\G_\post}(n) + 1\ :\ f_{\G}(n)]} = \post(n)\big] = \sum_{z \in \{0,1\}^{f_{\G_\post}(n)}} \Pr[\G(x) = z.\post(n)] \neq 0.
$$
The former condition ensures that the \emph{output size of $\G_\post$, $f_{\G_\post}$,} is always positive, and the latter condition ensures that the probability that $\G$ outputs a string whose last $|\post(n)|$ bits equal $\post(n)$ is nonzero. This way, the probability the last $|\post(n)|$ registers of $\G(x)$ have $\ket{\post(n)}$ as the state is non-zero, which is necessary for the probability that $\G_\post$ outputs $y \in \{0,1\}^{f_{\G_\post}(n)}$ on input $x \in \{0,1\}^n$ to be well-defined:
\begin{align*}
\Pr\left[\G_\post(x) = y \right] &\defeq \Pr\big[\G(x)_{[1\ :\ f_{\G_\post}(n)]} = y \mid \G(x)_{[f_{\G_\post}(n) + 1\ :\ f_{\G}(n)]} = \post(n)\big]\\
&= \frac{\Pr\left[\G(x) = y.\post(n)\right]}{\Pr\big[\G(x)_{[f_{\G_\post}(n) + 1\ :\ f_{\G}(n)]} = \post(n)\big]}.
\end{align*}
Notice, this is the conditional probability that the upper $f_{\G_\post}(n)$ registers of $\G(x)$ have $\ket{y}$ as the state, given that the lower $|\post(n)|$ registers of $\G(x)$ have $\ket{\post(n)}$ as the state.
\end{definition}

We now introduce the complexity class $\GadBQP(\S)$, which characterizes the languages decidable by efficient post-selected gadget quantum computers over the gate set $\S$.
\begin{definition}
Let $\S$ be a gate set. The class $\GadBQP(\S)$ consists of all languages $L$ for which there is an efficient post-selected gadget quantum computer $\G_\post$ over $\S$ and $\delta \in (0,1/2)$ such that for all $x \in \{0,1\}^*$, $\Pr[\G_\post(x)_1 = L(x)] \geq 1/2 + \delta.$
\end{definition}

It is plain that for every gate set $\S$, $\PostBQP(\S) \subseteq \GadBQP(\S) \subseteq \PostBQP.$ Thus, if $\S$ is universal, then $\PostBQP(\S) = \GadBQP(\S) = \PostBQP$. In this case, $\GadBQP(\S)$ is independent of $\delta$ for the same reason that $\PostBQP$ is. However, for non-universal $\S$, $\GadBQP(\S)$ is not necessarily independent of $\delta$, because $\S$ may be sufficiently constrained that no gadget quantum computer over $\S$ can approximate the majority function in any useful sense. Therefore, a priori, if $L \in \GadBQP(\S)$ for some non-universal $\S$, then there merely \emph{exists} $\delta \in (0, 1/2)$ and an efficient post-selected gadget quantum computer $\G_\post$ that decides $L$ to within bounded error $1/2 + \delta$.

To better understand the relationship between $\PostBQP(\S)$ and $\GadBQP(\S)$, observe that every efficient gadget quantum computer is an efficient post-selected quantum computer over the same gate set (simply substitute every gadget for its corresponding element of $\Gamma$). Therefore, every efficient \emph{post-selected} gadget quantum computer is an efficient post-selected quantum computer. This observation implies the following proposition.

\begin{proposition}
For every gate set $\S$, $\PostBQP(\S) = \GadBQP(\S) \subseteq \PostBQP$.
\label{prop:mainprop}
\end{proposition}

Ultimately, we are interested in those \emph{non}-universal gate sets $\S$ for which $\GadBQP(\S) = \PostBQP$. In this case, $\GadBQP(\S)$ becomes independent of $\delta$, so that if $L \in \PostBQP(\S) = \GadBQP(\S)$, then \emph{for all} $\delta \in (0, 1/2)$, there exists an efficient gadget quantum computer that decides $L$ to within bounded error $1/2 + \delta$. This fact is paramount to \propref{prop:hardness} below.

\section{Statement and Proof of Main Result}
\label{sec:mainresults}

Our main result affords a sufficient condition for $\GadBQP(\S) = \PostBQP$ for a non-universal gate set $\S$. Together with the following proposition and the reasonable assumption that the polynomial hierarchy does not collapse, we obtain a sufficient condition for determining if efficient classical computers cannot simulate efficient quantum computers over a non-universal gate set $\S$ in the weak multiplicative sense.

\begin{proposition}
\label{prop:hardness}
Let $\S$ be a gate set. If $\GadBQP(\S) = \PostBQP$ and for every efficient quantum computer $\Q$ over $\S$ there exists an efficient classical computer $\C$ that simulates $\Q$ to within multiplicative error $\epsilon < \sqrt{2} - 1$, then $\PH \subseteq \Sigma_3\P$. Therefore, if the polynomial hierarchy is infinite and $\GadBQP(\S) = \PostBQP$, then efficient classical computers cannot simulate efficient quantum computers over $\S$ to within multiplicative error $\epsilon < \sqrt{2} - 1$.
\end{proposition}
\begin{proof}
Aaronson proved that $\PostBQP = \PP$ \cite{Aar05}, Toda proved that $\PH \subseteq \P^\PP$ \cite{Tod91}, and Han, Hemaspaandra, and Thierauf proved that $\PostBPP \subseteq \P^{\Sigma_2\P}$ \cite{HHT97}. Therefore, if $\PostBPP = \PostBQP$, then $\PH \subseteq \P^{\PostBPP} \subseteq \P^{\P^{\Sigma_2\P}} \subseteq \Sigma_3\P$. Since $\PostBPP \subseteq \PostBQP$ unconditionally, it suffices to show that our additional assumptions imply $\PostBQP \subseteq \PostBPP$.

To this end, consider any $L \in \PostBQP$. By the assumption $\GadBQP(\S) = \PostBQP$ and \propref{prop:mainprop}, it holds that $\PostBQP(\S) = \PostBQP$. Thus, $L \in \PostBQP(\S)$, so there exists an efficient post-selected quantum computer $\Q_\post = (\Q, \post) = (Q, \S, \anc, \post)$ that decides $L$ to within bounded error $\delta \in (0, 1/2)$. Since $L \in \PostBQP$, we can choose any $\delta$ in this range. 

By assumption, there is an efficient classical computer $\C$ that simulates $\Q$ to within multiplicative error $\epsilon < \sqrt{2} - 1$. Therefore, for all $n \in \NN$, all $x \in \{0,1\}^n$, and all $y \in \{0,1\}^{f_{\Q_\post}(n)}$,
$$
\frac{1}{1 + \epsilon} \Pr[\C(x) = y.\post(n)] \leq \Pr[\Q(x) = y.\post(n)] \leq (1 + \epsilon) \Pr[\C(x) = y.\post(n)].
$$
Consequently,
\begin{align*}
\Pr[\Q_\post(x) = y] &= \frac{\Pr\left[\Q(x) = y.\post(n)\right]}{\Pr\big[\Q(x)_{[f_{\Q_\post}(n) + 1\ :\ f_\Q(n)]} = \post(n)\big]}\\
& \leq \frac{(1 + \epsilon) \Pr\left[\C(x) = y.\post(n)\right]}{\Pr\big[\C(x)_{[f_{\C_\post}(n) + 1\ :\ f_\C(n)]} = \post(n)\big] / (1 + \epsilon)}\\
& = (1 + \epsilon)^2\Pr[\C_\post(x) = y],
\end{align*}
where $\C_\post$ is the efficient post-selected classical computer $(\C, \post)$. Thus, for every $x \in \{0,1\}^*$,
$$
\Pr[\C_\post(x)_1 = L(x)] \geq \frac{1}{(1 + \epsilon)^2} \left(\frac{1}{2} + \delta\right) = \frac{1}{2} \cdot \frac{1 + 2\delta}{(1 + \epsilon)^2}.
$$
Consequently, $\C_\post$ decides $L$ in the sense of $\PostBPP$ provided $(1 + \epsilon)^2 < 1 + 2\delta$. Since $\delta$ can be any value satisfying $0 < \delta < 1/2$, it suffices for $0 \leq \epsilon < \sqrt{2} - 1$ to ensure $L \in \PostBPP$.
\end{proof}

We now state and prove our main technical result, which gives a sufficient condition for a gate set $\S$ to satisfy $\GadBQP(\S) = \PostBQP$.

\begin{theorem}
\label{thm:mainone}
Let $\S$ be a gate set. If there exists a finite subset $\Gamma \subset \gad_1(\S)$ that densely generates $\SL(2;\CC)$, then $\GadBQP(\S) = \PostBQP$.
\end{theorem}

To prove this, we rely on a result due to Bouland and Giurgic\u a-Tiron \cite{BG21}, which establishes an inverse-free version of the non-unitary Solovay-Kitaev theorem of Aharanov, Arad, Eban, and Landau \cite{AAEL08}. In the following, $\opnorm{\cdot}$ is the operator norm.

\begin{theorem}[Corollary of \cite{BG21}]
Let $\Gamma \subset \SL(2;\CC)$ be a finite subset that densely generates $\SL(2;\CC)$. There exists a constant $c > 0$ such that for all $\delta > 0$ and all $\omega \in \SL(2;\CC)$, there is a sequence $\sigma_\omega$ of operators from $\Gamma$ of length $O(\log^c 1 / \delta)$ such that $\opnorm{\omega - \sigma_\omega} < \delta$. Moreover, there is an algorithm that constructs $\sigma_\omega$ in polylogarithmic time in the parameter $1 / \delta$.
\label{thm:nonunitarySK}
\end{theorem}

This implies \thmref{thm:mainone}.

\begin{proof}[Proof of \thmref{thm:mainone}]
By \propref{prop:mainprop}, $\GadBQP(\S) \subseteq \PostBQP$, so it remains to show the reverse containment. To this end, let $E$ be an entangling gate in $\S$, which exists by our definition of ``gate set'' (see \defref{def:gateset}). Then, $\S' = \{E, \tilde{H}, \tilde{T}\}$ is a universal gate set, where 
$$
\tilde{H} = 
\frac{i}{\sqrt{2}}
\begin{pmatrix}
1 & 1\\
1 & -1
\end{pmatrix}
\quad \text{and} \quad
\tilde{T} = 
\begin{pmatrix}
e^{-i \pi / 8} & 0\\
0 & e^{i \pi / 8}
\end{pmatrix}
$$
are the $\SU(2)$ versions of the usual Hadamard and $\pi/8$ gates, respectively. Consequently, $\PostBQP(\S') = \PostBQP$, so $\PostBQP \subseteq \GadBQP(\S)$ iff $\PostBQP(\S') \subseteq \GadBQP(\S)$.

To prove the latter containment, let $L \in \PostBQP(\S')$. Then, there is an efficient post-selected quantum computer $(Q' = (Q'_n)_{n \in \NN}, \S', \anc', \post')$ of size $s'(n)$ that decides $L$. By assumption, there exists a finite subset $\Gamma \subset \gad_1(\S)$ that densely generates $\SL(2;\CC)$. Thus, by \thmref{thm:nonunitarySK}, there exists $c > 0$ such that for all $n \in \NN$, there are sequences $\sigma_{\tilde{H}}(n)$ and $\sigma_{\tilde{T}}(n)$ of gates from $\Gamma$ with lengths $O(\log^{c} (s'(n)2^{n}))$ such that $\opnorm{\tilde{H} - \sigma_{\tilde{H}}(n)} < 2^{-n} / s'(n)$ and $\opnorm{\tilde{T} - \sigma_{\tilde{T}}(n)} < 2^{-n} / s'(n)$. By \thmref{thm:nonunitarySK}, both $\sigma_{\tilde{H}}(n)$ and $\sigma_{\tilde{T}}(n)$ can be constructed from $\Gamma$ in polylogarithmic time $s'(n) 2^n$. Therefore, both $\sigma_{\tilde{H}}(n)$ and $\sigma_{\tilde{T}}(n)$ can be constructed from $\Gamma$ in polynomial time with respect to the parameter $n$. Consequently, for each $n \in \NN$, replacing every $\tilde{H}$ and $\tilde{T}$ gate in the $(n + |\anc'(n)|)$-qubit quantum circuit $Q_n'$ with $\sigma_{\tilde{H}}(n)$ and $\sigma_{\tilde{T}}(n)$, respectively, yields an $(n + |\anc'(n)|)$-qubit gadget quantum circuit $Q_n$ over $\S \cup \Gamma$ such that $\opnorm{Q'_n - Q_n} < 2^{-n}$. Thus, $Q \defeq (Q_n)_{n \in \NN}$ is a uniform and polynomial size family of $(n + |\anc'(n)|)$-qubit gadget quantum circuits over $\S \cup \Gamma$ that approximates $Q'$ to exponential accuracy. It holds, therefore, that $(Q, \S, \Gamma, \anc', \post')$ is an efficient post-selected gadget quantum computer over $\S$ that decides $L$, so $L \in \GadBQP(\S)$.
\end{proof}

Consequently, if a finite number of single-qubit post-selection gadgets $\Gamma$ over a gate set $\S$ densely generates $\SL(2;\CC)$, then efficient quantum computers over $\S$ can perform a sampling task that no efficient classical computer can. The problem, then, is to understand when $\Gamma$ generates a dense subset of $\SL(2;\CC)$.

To do this, we study the restricted problem of understanding when $\Gamma$ generates a dense \emph{subgroup} of $\SL(2;\CC)$. Understanding this turns out to be algorithmically straightforward thanks to a result of Sullivan \cite{Sul85} (our \thmref{thm:SullivanTheorem}), which proves that the so-called ``non-elementary'' subgroups of $\SL(2;\CC)$ are essentially either discrete or dense in $\SL(2;\CC)$. Note that the theorem below is actually a simplified version of Sullivan's theorem, which we explain in more detail in \sref{sec:loxodromic}. Note also that while Sullivan's theorem is couched in several hitherto undefined terms, the point to keep in mind is that each premise is algorithmically easy to check.

\begin{restatable}[Simplified version of \thmref{thm:SullivanTheorem}]{theorem}{sullivan}
\label{thm:sullivanone}
Let $\Gamma$ be a finite subset of $\SL(2;\CC)$. If $\words{\Gamma}$ is a non-elementary, non-discrete, and strictly loxodromic subgroup of $\SL(2;\CC)$, then $\words{\Gamma}$ is dense in $\SL(2;\CC)$.
\end{restatable}

Together, \propref{prop:hardness}, \thmref{thm:mainone}, and \thmref{thm:sullivanone} entail our criterion for concluding if no efficient classical computer can simulate an efficient quantum computer over a non-universal gate set to within small multiplicative error:

\begin{theorem}[Criterion for Post-Selected Quantum Advantage]
\label{thm:maintwo}
Let $\S$ be a gate set and suppose the polynomial hierarchy is infinite. If there exists a finite subset $\Gamma \subset \gad_1(\S)$ such that $\words{\Gamma}$ is a non-elementary, non-discrete, and strictly loxodromic subgroup of $\SL(2;\CC)$, then efficient classical computers cannot simulate efficient quantum computers over $\S$ to within multiplicative error $\epsilon < \sqrt{2} - 1$.
\end{theorem}

Of course, since we have yet to define what the terms ``non-elementary'' and ``strictly loxodromic'' mean, it is not apparent if our criterion is actually of any use. However, as we detail in the next section, these terms are not complicated, and there is in fact a simple algorithm that checks if a finitely generated subgroup $\words{\Gamma}$ of $\SL(2;\CC)$ is non-elementary, non-discrete, and strictly loxodromic (and hence dense in $\SL(2;\CC)$ {\`a} la \thmref{thm:sullivanone}). Consequently, on account of \thmref{thm:maintwo}, there is a simple algorithm that checks if efficient classical computers cannot simulate efficient quantum computers over a non-universal gate set $\S$ in the weak multiplicative sense. In \sref{sec:applications}, we illustrate the power of our criterion by proving several new quantum advantage results in addition to some well-known old ones.

We emphasize that this approach provides a \textit{sufficient} criterion to infer the classical intractability of a quantum gate set, but not a strictly necessary one. In other words, it could happen that a non-universal gate set does not satisfy the premises of \thmref{thm:maintwo}, but nevertheless engenders circuits that no efficient classical computer can simulate.\footnote{This could happen, for example, if $\Gamma$ generates a dense subgroup of $\SU(2)$ as opposed to $\SL(2;\CC)$.} There are also a few ways we think our criterion \thmref{thm:maintwo} can be improved; see Conjectures~\ref{conj:gadgetconj} and \ref{conj:densityresult}.

\section{Some Properties of Subgroups of $\SL(2;\CC)$}
\label{sec:grouptheory}

There are several important properties that a subgroup of $\SL(2;\CC)$ can possess. The first property is \emph{elementarity}. 

\subsection{Elementary Subgroups of $\SL(2;\CC)$}

Given an element $x$ of a set $X$ and a group $H$ that acts on $X$ from the left via the operation $\alpha : H \times X \rightarrow X$, recall that the \emph{$H$-orbit of $x$} is the set $\{\alpha(h,x) \mid h \in H\}$. We say there exists a \emph{finite $H$-orbit in $X$} iff there is $x \in X$ such that the $H$-orbit of $x$ is a finite set.

\begin{definition}
A subgroup $H \leq \SL(2;\CC)$ is \emph{elementary} iff there exists a finite $H$-orbit in $\RR^3$.
\end{definition}

For example, $\SU(2) \leq \SL(2;\CC)$ is elementary. To see this, consider the representation $\rho \defeq R \circ \pi$, where $\pi : \SU(2) \rightarrow \SO(3) \cong \SU(2) / \mathbb{Z}_2$ is the covering homomorphism and $R : \SO(3) \rightarrow \GL(3; \RR)$ is the defining (i.e. vector) representation of $\SO(3)$. Since every rotation in $\RR^3$ fixes the origin $(0,0,0)^T \in \RR^3$, $\{\rho(U) (0,0,0)^T \mid U \in \SU(2)\}$ is a finite $\SU(2)$-orbit in $\RR^3$.

Elementary subgroups are important in the theory of M\"obius transformations, for which a standard introduction is a textbook by Bearden \cite{Bea83}. Rather than digressing into this theory, however, we shall simply quote a series of useful results for determining when a subgroup of $\SL(2;\CC)$ is elementary. The first of these is an exercise in Bearden's textbook.

\begin{proposition}[Exercise 5.1, Problem 2 in \cite{Bea83}]
A subgroup $H \leq \SL(2;\CC)$ is elementary iff for all $g, h \in H$, the two-generated subgroup $\words{g, h}$ is elementary.
\label{prop:elemprop}
\end{proposition}

Therefore, the elementarity of a group depends solely on the information contained in all of its rank-two subgroups. Unfortunately, if $H \leq \SL(2;\CC)$ is finitely generated, then it is not necessarily the case that $H$ is elementary if every two-generated subgroup of generators is elementary.\footnote{\label{foot:one}A counterexample follows from the fact that any two involutions generate an elementary subgroup \cite{Koh24}.} Nevertheless, the logical inverse of this statement is true by \propref{prop:elemprop}, so there is a sufficient, purely generator-based condition for a finitely generated subgroup of $\SL(2;\CC)$ to be \emph{non}-elementary:

\begin{corollary}
Let $H = \words{\Gamma} \leq \SL(2;\CC)$ be finitely generated by $\Gamma \subset \SL(2;\CC)$. If there exist $g, h \in \Gamma$ such that $\words{g,h}$ is non-elementary, then $H$ is non-elementary.
\label{cor:elemcor}
\end{corollary}

The next proposition by Baribeau and Ransford \cite{BR00} entails a very useful necessary and sufficient condition for a two-generated subgroup of $\SL(2;\CC)$ to be elementary.
\begin{proposition}[Proposition 2.1 in \cite{BR00}]\label{thm:Baribeau}
For all $g, h \in \SL(2;\CC)$, define $\beta(g) \defeq \tr^2(g) - 4$ and $\gamma(g,h) \defeq \tr(ghg^{-1}h^{-1}) - 2$. Then, $\words{g, h}$ is an elementary subgroup of $\SL(2;\CC)$ iff one of the following three conditions hold:
\begin{enumerate}[(i)]
\item $\beta(g), \beta(h) \in [-4,0]$, and $\gamma(g,h) \in [-\beta(g)\beta(h) / 4, 0]$,
\item $\gamma(g,h) = 0$,
\item $\beta(g) = \gamma(g,h)$ and $\beta(h) = -4$, or $\beta(g) = -4$ and $\beta(h) = \gamma(g,h)$, or $\beta(g) = -4$ and $\beta(h) = -4$.
\end{enumerate}
\label{prop:elementarycondition}
\end{proposition}

Altogether, one gets \algref{alg:IsElementary}, $\IsElementary$, for determining if a finitely generated subgroup of $\SL(2;\CC)$ is non-elementary. This algorithm combines \corref{cor:elemcor} with the contrapositive of \propref{prop:elementarycondition}.

\begin{algorithm}
\caption{$\IsElementary$}
\label{alg:IsElementary}
\begin{algorithmic}[1]
\STATE \textbf{Input:} Finite set $\Gamma \subset \SL(2;\CC)$ such that $\words{\Gamma} \leq \SL(2;\CC)$
\STATE \textbf{Output:} $\NO$ if $\words{\Gamma}$ is non-elementary, $\IDK$ (``I don't know'') otherwise
\FOR {$g,h \in \Gamma$}
		\IF {$\beta(g) \not\in [-4,0]$ or $\beta(h) \not\in [-4,0]$ or $\gamma(g,h) \not\in [-\beta(g)\beta(h) / 4, 0]$}
            \IF {$\gamma(g,h) \neq 0$}
			\IF {$\beta(g) \neq \gamma(g,h)$ or $\beta(h) \neq -4$}
				\IF {$\beta(g) \neq -4$ or $\beta(h) \neq \gamma(g,h)$}
					\IF {$\beta(g) \neq -4$ or $\beta(h) \neq -4$}	
						\STATE return $\NO$
					\ENDIF
     \ENDIF
				\ENDIF
			\ENDIF
		\ENDIF
\ENDFOR
\STATE return $\IDK$
\end{algorithmic}
\end{algorithm}

\subsection{Discrete Subgroups of $\SL(2;\CC)$}

Another important property that a subgroup of $\SL(2;\CC)$ can possess is \emph{discreteness}. This property is the familiar topological one, so we do not define it here. Instead, we cite several famous results that help determine when a non-elementary subgroup of $\SL(2;\CC)$ is discrete.

Like elementarity, the first result proves that the discreteness of a non-elementary subgroup of $\SL(2;\CC)$ depends solely on the information contained in all of its rank-two subgroups.

\begin{proposition}[Theorem 5.4.2 in \cite{Bea83}]
A non-elementary subgroup $H \leq \SL(2;\CC)$ is discrete iff for all $g, h \in H$, the two-generated subgroup $\words{g, h}$ is discrete.
\label{prop:discreteprop}
\end{proposition}

Unfortunately, if $H \leq \SL(2;\CC)$ is finitely generated and non-elementary, then it is not necessarily the case that $H$ is discrete if every two-generated subgroup of generators is discrete.\footnote{\label{foot:two}A counterexample follows from the following three facts \cite{Koh24}: (1) there are arbitrary small triangles $T$ in the hyperbolic plane $\mathbb{H}^2$ with angles that are rational multiples of $\pi$, (2) the hyperbolic isometry group $G$ generated by reflections in the edges of $T$ is non-discrete if $T$ is small enough, and (3) any two generators of $G$ generate a finite group of isometries.} Nevertheless, the logical inverse of this statement is true by \propref{prop:discreteprop}, so there is a sufficient, purely generator-based condition for a finitely generated and non-elementary subgroup of $\SL(2;\CC)$ to be \emph{non}-discrete:
\begin{corollary}
\label{cor:nondiscrete}
Let $H = \words{\Gamma} \leq \SL(2;\CC)$ be non-elementary and finitely generated by $\Gamma \subset \SL(2;\CC)$. If there exist $g, h \in \Gamma$ such that $\words{g, h}$ is non-discrete, then $H$ is non-discrete. 
\end{corollary}

Consequently, to determine if a finitely generated and non-elementary subgroup $H \leq \SL(2;\CC)$ is non-discrete, it suffices to interrogate the rank-two subgroups of $H$ that are generated by the generators. The most natural way to do this is to employ a famous result of J{\o}rgensen \cite{Jor76}, which concerns every discrete and non-elementary two-generated subgroup of $\SL(2;\CC)$.

\begin{proposition}[J{\o}rgensen's Inequality \cite{Jor76}]
\label{prop:Jorgensen}
Suppose $\words{g,h}$ generates a discrete and non-elementary subgroup of $\SL(2;\CC)$. Then 
\begin{equation}
\left| \tr^2(g) - 4 \right| + \left| \tr(ghg^{-1}h^{-1}) - 2 \right| \geq 1.
\label{eq:jorg}
\end{equation}
Therefore, if $\words{g,h}$ is a non-elementary subgroup of $\SL(2;\CC)$ and $g$ and $h$ violate \eref{eq:jorg}, then $\words{g,h}$ is non-discrete.
\end{proposition}

In fact, there are many interesting generalizations of J{\o}rgensen's inequality, such as the following few. For more, see \cite{BM81, CT09, AGM21}.

\begin{proposition}[Theorems 1 -- 3 in \cite{Tan89}]
\label{prop:Jorg_generalized}
Suppose $\words{g,h}$ generates a discrete subgroup of $\SL(2;\CC)$. 
\begin{enumerate}[(i)]
\item If $\tr(ghg^{-1}h^{-1}) \neq 1$, then 
$$
\left| \tr^2(g) - 2 \right| + \left| \tr(ghg^{-1}h^{-1}) - 1 \right| \geq 1.
$$
\item If $\tr(ghg^{-1}h^{-1}) = 1$ and $\tr^2(g) \neq 2$, then 
$$
\left| \tr^2(g)  - 2 \right| > \frac{1}{2}.
$$
\item If $\tr^2(g) \neq 1$, then 
$$
\left| \tr^2(g) - 1 \right| + \left| \tr(ghg^{-1}h^{-1}) \right| \geq 1.
$$
\item If $\tr^2(g) = 1$, then 
$$
\left| \tr(ghg^{-1}h^{-1}) \right| > \frac{1}{2} \quad \text{or} \quad \tr(ghg^{-1}h^{-1}) = 0
$$ 
and
$$
\left| \tr(ghg^{-1}h^{-1}) - 1 \right| > \frac{1}{2} \quad \text{or} \quad \tr(ghg^{-1}h^{-1}) = 1.
$$
\item If $\tr(ghg^{-1}h^{-1}) \neq 1$, then 
$$
\left| \tr^2(g) - \tr(ghg^{-1}h^{-1}) \right| + \left| \tr(ghg^{-1}h^{-1}) - 1 \right| \geq 1.
$$
\item If $\tr(ghg^{-1}h^{-1}) = 1$ and $\tr^2(g) \neq 1$, then 
$$
\left| \tr^2(g)  - 1 \right| > \frac{1}{2}.
$$
\end{enumerate}
\end{proposition}

Altogether, one gets \algref{alg:discrete}, $\IsDiscrete$, for determining if a finitely generated and non-elementary subgroup of $\SL(2;\CC)$ is non-discrete. This algorithm combines \corref{cor:nondiscrete} with the contrapositive of J{\o}rgensen's inequality (\propref{prop:Jorgensen}) and the contrapositive of \propref{prop:Jorg_generalized}.

\begin{algorithm}
\caption{$\IsDiscrete$}
\label{alg:discrete}
\begin{algorithmic}[1]
\STATE \textbf{Input:} Finite set $\Gamma \subset \SL(2;\CC)$ such that $\words{\Gamma} \leq \SL(2;\CC)$
\STATE \textbf{Output:} $\NO$ if $\words{\Gamma}$ is non-discrete, $\IDK$ otherwise
\FOR {$g, h \in \Gamma$}
        \IF {$\IsElementary(\{g,h\}) = \NO$ and $\left| \tr^2(g) - 4 \right| + \left| \tr(ghg^{-1}h^{-1}) - 2 \right| < 1$}
		  \STATE return $\NO$
        \ENDIF
        \IF {$\tr(ghg^{-1}h^{-1}) \neq 1$ and $\left| \tr^2(g) - 2 \right| + \left| \tr(ghg^{-1}h^{-1}) - 1 \right| < 1$}
            \STATE return $\NO$
        \ENDIF
        \IF {$\tr(ghg^{-1}h^{-1}) = 1$ and $\tr^2(g) \neq 2$ and $\left| \tr^2(g)  - 2 \right| \leq 1/2$}
            \STATE return $\NO$
        \ENDIF
        \IF {$\tr^2(g) \neq 1$ and $\left|\tr^2(g) - 1\right| + \left|\tr(ghg^{-1}h^{-1})\right| < 1$}
            \STATE return $\NO$
        \ENDIF
        \IF {$\tr^2(g) = 1$ and $\left|\tr(ghg^{-1}h^{-1})\right| \leq 1/2$ and $\tr(ghg^{-1}h^{-1}) \neq 0$}
            \STATE return $\NO$
        \ENDIF
        \IF {$\tr^2(g) = 1$ and $\left|\tr(ghg^{-1}h^{-1}) - 1\right| \leq 1/2$ and $\tr(ghg^{-1}h^{-1}) \neq 1$}
            \STATE return $\NO$
        \ENDIF
        \IF {$\tr(ghg^{-1}h^{-1}) \neq 1$ and $\left|\tr^2(g) - \tr(ghg^{-1}h^{-1})\right| + \left|\tr(ghg^{-1}h^{-1}) - 1\right| < 1$}
            \STATE return $\NO$
        \ENDIF
        \IF {$\tr(ghg^{-1}h^{-1}) = 1$ and $\tr^2(g) \neq 1$ and $\left|\tr^2(g) - 1\right| \leq 1/2$}
            \STATE return $\NO$
        \ENDIF
\ENDFOR
\STATE return $\IDK$
\end{algorithmic}
\end{algorithm}

\subsection{Strictly Loxodromic Subgroups of $\SL(2;\CC)$}
\label{sec:loxodromic}

Another important property that a subgroup of $\SL(2;\CC)$ can possess is \emph{strict loxodromy}. As no subgroup of $\SL(2;\RR)$ is strictly loxodromic, this property is a purely complex phenomenon.

\begin{definition}
A subgroup $H \leq \SL(2;\CC)$ is \emph{loxodromic} iff there exists $g \in H$ such that $\tr^2(g) \in \CC \backslash [0,4]$. We say $H \leq \SL(2;\CC)$ is \emph{strictly loxodromic} iff there exists $g \in H$ such that $\tr(g) \in \CC \backslash \RR$. In this latter case, $g$ is called a \emph{strictly loxodromic element} of $\SL(2;\CC)$.
\end{definition}

Unfortunately, like discreteness and elementarity, if $H \leq \SL(2;\CC)$ is finitely generated, then it is not necessarily the case that $H$ is strictly loxodromic only if there is a strictly loxodromic generator.\footnote{A counterexample is the group $H = \words{\{\omega_1, \omega_2, \omega_1^{-1}, \omega_2^{-1}\}}$, where
$$
\omega_1 =
\begin{pmatrix}
i & 0\\
0 & -i
\end{pmatrix}
\quad \text{and} \quad
\omega_2 =
\begin{pmatrix}
1 & 2\\
-1 & -1
\end{pmatrix}.
$$
Evidently, $H$ is strictly loxodromic because $\tr(\omega_1\omega_2) \in \CC \backslash \RR$, however no generator of $H$ is strictly loxodromic.} Nevertheless, the converse of this statement is true definitionally:

\begin{lemma}
\label{lem:loxlem}
Let $H = \words{\Gamma} \leq \SL(2;\CC)$ be finitely generated by $\Gamma \subset \SL(2;\CC)$. If $\Gamma$ contains a strictly loxodromic element, then $H$ is strictly loxodromic.
\end{lemma}

Altogether, \lemref{lem:loxlem} entails the very simple \algref{alg:lox}, $\IsLoxodromic$, for determining if a finitely generated subgroup of $\SL(2;\CC)$ is strictly loxodromic.

\begin{algorithm}
\caption{$\IsLoxodromic$}
\label{alg:lox}
\begin{algorithmic}[1]
\STATE \textbf{Input:} Finite set $\Gamma \subset \SL(2;\CC)$ such that $\words{\Gamma} \leq \SL(2;\CC)$
\STATE \textbf{Output:} $\YES$ if $\words{\Gamma}$ is strictly loxodromic, $\IDK$ otherwise
\FOR {$g \in \Gamma$}
	\IF {$\tr(g) \in \CC \backslash \RR$}
		\STATE return $\YES$
	\ENDIF
\ENDFOR
\STATE return $\IDK$
\end{algorithmic}
\end{algorithm}

\subsection{Dense Subgroups of $\SL(2;\CC)$}

If $H \leq \SL(2;\CC)$ is non-discrete, then it may be topologically dense in $\SL(2;\CC)$. Surprisingly, at least for non-elementary $H$, this is essentially always the case, according to a result of Sullivan \cite{Sul85} (c.f. Theorem 9.3 in \cite{AAEL08}).

\begin{theorem}[First proposition in \cite{Sul85}]
\label{thm:SullivanTheorem}
If $H \leq \SL(2;\CC)$ is non-elementary and non-discrete, then $H$ is either dense in $\SL(2;\CC)$ or conjugate to a dense subgroup of $\SL(2;\RR)$.
\end{theorem}

If $H$ is conjugate to a dense subgroup of $\SL(2;\RR)$, then there exists $\omega \in \SL(2;\CC)$ such that $\bar{H}$, the topological closure of $H$, equals $\omega^{-1} \SL(2;\RR) \omega$. While we conjecture that density in this sense should yield a statement similar to \thmref{thm:mainone} (see \conjref{conj:densityresult}), we were not able to prove this. Instead, we use strict loxodromy to restrict $H$ enough so that it is necessarily dense in $\SL(2;\CC)$. This was the content of \thmref{thm:sullivanone},  which we restate and prove below.

\sullivan*

\begin{proof}
For all $\omega \in \SL(2;\CC)$ and all $g \in \omega^{-1} \SL(2;\RR) \omega$, $\tr(g) \in \RR$. Thus, $\omega^{-1} \SL(2;\RR) \omega$ contains no strictly loxodromic elements, so $\words{\Gamma}$ must be dense in $\SL(2;\CC)$ by \thmref{thm:SullivanTheorem}.
\end{proof}

By composing the algorithms $\IsElementary$, $\IsDiscrete,$ and $\IsLoxodromic$, we obtain an algorithm for checking if a finitely generated subgroup of $\SL(2;\CC)$ is dense in $\SL(2;\CC)$. The exact composition is illustrated on the left side of \figref{fig:criterionschematic}. Put in the context of our criterion \thmref{thm:maintwo}, we get an algorithm for checking if efficient classical computers cannot simulate efficient quantum computers over a non-universal gate set to within small multiplicative error (assuming the polynomial hierarchy is infinite). This is depicted on the right side of \figref{fig:criterionschematic}.

\begin{figure}
\begin{subfigure}{0.49\textwidth}
\begin{tikzpicture}[>=latex']
        \tikzset{block/.style= {draw, rectangle, align=center,minimum width=2cm,minimum height=0.75cm},
        }
        \node (start) {$\Gamma \subset \SL(2;\CC)$};
        \node [block, below = 1cm of start] (ELEM){$\IsElementary$};
        \node [coordinate, below = 1cm of ELEM] (c2) {};
        \node [block, below right = 1cm and 0.3cm of c2] (JORG){$\IsDiscrete$};
        \node [below left = 1cm and 0.3cm of c2] (case2){unclear};
        \node [coordinate, below = 1cm of JORG] (c3) {};
        \node [block, below right = 1cm and 0.2cm of c3] (LOX){$\IsLoxodromic$};
        \node [below left = 1cm and 0.3cm of c3] (case3){unclear};
        \node [coordinate, below = 1cm of LOX] (c4) {};
        \node [below right = 1cm and 0.6cm of c4] (case5){$\begin{matrix} \text{$\words{\Gamma}$ dense}\\ \text{in $\SL(2;\CC)$}\end{matrix}$};
        \node [below left = 1cm and 0.3cm of c4] (case4){unclear};
        
        \draw[-stealth] (start) edge (ELEM);
        \draw[-] (ELEM) -- (c2);
        \path[-stealth]
        (c2) edge node[sloped, xshift=5pt, yshift=5pt] {$\IDK\,\,\,\,$} (case2)
        (c2) edge node[sloped, xshift=5pt, yshift=5pt] {$\NO\quad$} (JORG);
	\draw[-] (JORG) -- (c3);
	\path[-stealth]
        (c3) edge node[sloped, xshift=5pt, yshift=5pt] {$\IDK\,\,\,\,$} (case3)
        (c3) edge node[sloped, xshift=5pt, yshift=5pt] {$\NO\quad$} (LOX);
        	\draw[-] (LOX) -- (c4);
	\path[-stealth]
        (c4) edge node[sloped, xshift=5pt, yshift=5pt] {$\IDK\,\,\,\,$} (case4)
        (c4) edge node[sloped, xshift=5pt, yshift=5pt] {$\YES\quad$} (case5);
    \end{tikzpicture}
\end{subfigure}
\hspace*{\fill}
\begin{subfigure}{0.49\textwidth}
\begin{tikzpicture}[>=latex']
        \tikzset{block/.style= {draw, rectangle, align=center,minimum width=2cm,minimum height=0.75cm},
        }
        \node (start) {$\Gamma \subset \gad_1(\S)$};
        \node [block, below = 1cm of start] (ELEM){$\IsElementary$};
        \node [coordinate, below = 1cm of ELEM] (c2) {};
        \node [block, below right = 1cm and 0.3cm of c2] (JORG){$\IsDiscrete$};
        \node [below left = 1cm and 0.3cm of c2] (case2){unclear};
        \node [coordinate, below = 1cm of JORG] (c3) {};
        \node [block, below right = 1cm and 0.2cm of c3] (LOX){$\IsLoxodromic$};
        \node [below left = 1cm and 0.3cm of c3] (case3){unclear};
        \node [coordinate, below = 1cm of LOX] (c4) {};
        \node [below right = 1cm and 0.6cm of c4] (case5){$\begin{matrix} \text{classically}\\ \text{intractable}\end{matrix}$};
        \node [below left = 1cm and 0.3cm of c4] (case4){unclear};
        
        \draw[-stealth] (start) edge (ELEM);
        \draw[-] (ELEM) -- (c2);
        \path[-stealth]
        (c2) edge node[sloped, xshift=5pt, yshift=5pt] {$\IDK\,\,\,\,$} (case2)
        (c2) edge node[sloped, xshift=5pt, yshift=5pt] {$\NO\quad$} (JORG);
	\draw[-] (JORG) -- (c3);
	\path[-stealth]
        (c3) edge node[sloped, xshift=5pt, yshift=5pt] {$\IDK\,\,\,\,$} (case3)
        (c3) edge node[sloped, xshift=5pt, yshift=5pt] {$\NO\quad$} (LOX);
        	\draw[-] (LOX) -- (c4);
	\path[-stealth]
        (c4) edge node[sloped, xshift=5pt, yshift=5pt] {$\IDK\,\,\,\,$} (case4)
        (c4) edge node[sloped, xshift=5pt, yshift=5pt] {$\YES\quad$} (case5);
    \end{tikzpicture}
\end{subfigure}
\caption{(Left) A flowchart illustrating our algorithm for checking if a finite subset $\Gamma \subset \SL(2;\CC)$ generates a dense subgroup of $\SL(2;\CC)$. (Right) An identical flowchart, but placed in the context of our criterion, \thmref{thm:maintwo}. In this case, $\Gamma$ is a finite subset of $\gad_1(\S)$ for some gate set $\S$ and ``$\words{\Gamma}$ dense in $\SL(2;\CC)$'' is replaced by ``classically intractable'' in the sense of \thmref{thm:maintwo}.}
\label{fig:criterionschematic}
\end{figure}

\section{Applications}
\label{sec:applications}

In this section, we demonstrate the utility of our criterion by proving quantum advantage results for several types of restricted quantum computational models, including: instantaneous quantum polynomial (IQP) circuits, conjugated Clifford circuits (CCCs), \emph{commuting} CCCs, CCCs over various fragments of the Clifford group, and CCCs where the interstitial Clifford circuit is composed exclusively of $\CZ$ gates. Using our criterion, we re-derive the quantum advantage of IQP circuits and CCCs, which were originally found in \cite{BJS11} and \cite{BFK18}, respectively. The quantum advantage of commuting CCCs, CCCs over the Clifford fragments, and CCCs over just $\CZ$ gates are (as far as we know) new, and this speaks to the utility of our intractability criterion in settings where other standard approaches do not directly apply. 

In addition to these quantum advantage results (which are all statements of the form ``there exists a circuit in this restricted quantum model that no efficient classical computer can simulate unless the polynomial hierarchy collapses''), we also prove the full complexity classification of CCCs, commuting CCCs, and CCCs over some of the fragments of the Clifford group (which are all statements of the form ``if the polynomial hierarchy is infinite, then efficient classical computers can simulate circuits in this restricted quantum model if and only if the circuit takes such-and-such form''). These classification results make extensive use of Wolfram Mathematica 14.1. Our notebook is available online for anyone trying to reproduce our calculations \cite{KFG24}.

In what follows, we make extensive use of the following gates in the computational basis:
$$
H = \frac{1}{\sqrt{2}}
\begin{pmatrix}
1 & 1\\
1 & -1
\end{pmatrix},\quad
T = 
\begin{pmatrix}
1 & 0\\
0 & e^{i \pi / 4}
\end{pmatrix},\quad
S = 
\begin{pmatrix}
1 & 0\\
0 & i
\end{pmatrix},\quad \text{and} \quad
Z = 
\begin{pmatrix}
1 & 0\\
0 & -1
\end{pmatrix}.
$$ 
We also use the rotation matrices 
$$
R_x(\theta) = 
\begin{pmatrix}
\cos \frac{\theta}{2} & -i\sin \frac{\theta}{2}\\
-i\sin \frac{\theta}{2} & \cos \frac{\theta}{2}
\end{pmatrix}\quad \text{and} \quad
R_z(\theta) = 
\begin{pmatrix}
e^{-i \theta / 2} & 0\\
0 & e^{i \theta / 2}
\end{pmatrix}
$$
as well as the entangling gates
$$
\CZ = 
\begin{pmatrix}
1 & 0 & 0 & 0\\
0 & 1 & 0 & 0\\
0 & 0 & 1 & 0\\
0 & 0 & 0 & -1
\end{pmatrix}
\quad
\text{and}
\quad
\CNOT = 
\begin{pmatrix}
1 & 0 & 0 & 0\\
0 & 1 & 0 & 0\\
0 & 0 & 0 & 1\\
0 & 0 & 1 & 0
\end{pmatrix}.
$$

Finally, to help the reader navigate the structure of all of our quantum advantage and classification results, we have provided two dependency diagrams. One, \figref{fig:hardness_dependency}, indicates the statements that each quantum advantage result depends on, and the other, \figref{fig:classification_dependency}, indicates the statements that each classification result depends on.

\subsection{Instantaneous Quantum Polynomial Circuits}

At a high level, the IQP model describes quantum computations that only use commuting operations. Physically, this means that there is no temporal order to the computation, save the output measurement of the qubits, which happens last.

\begin{definition}
An \emph{IQP circuit} is an efficient quantum computer over a gate set $\S$ for which every $U \in \S$ is diagonal in the basis $\{\ket{0} \pm \ket{1}\}$. In other words, if $U \in \S$ is a $2^k \times 2^k$ matrix, then there exists $D \in \U(2^k)$ such that $D$ is diagonal in the computational basis and $U = H^{\otimes k} D H^{\otimes k}$.
\end{definition}

In \cite{BJS11}, Bremner, Jozsa, and Shepherd prove that when augmented with post-selection, $\IQP$ circuits can decide $\PP$-complete languages. Therefore, by an argument that is nearly identical to our \propref{prop:hardness}, one gets the following theorem.
\begin{theorem}[Corollary 1 in \cite{BJS11}]
\label{thm:IQP_hardness}
If the polynomial hierarchy is infinite, then efficient classical computers cannot simulate IQP circuits to within multiplicative error $\epsilon < \sqrt{2} - 1$.
\end{theorem}

Bremner et al. prove \thmref{thm:IQP_hardness} by first restricting to the gate set 
$$
\S_\IQP \defeq \{H T H, (H \otimes H) \CZ (H \otimes H)\}.
$$
Since $T$ and $\CZ$ are diagonal in the computational basis, circuits over $\S_\IQP$ are indeed IQP circuits. They then show that there exists a post-selection gadget over $\S_{\IQP}$ that can be used to ``inject" a Hadamard gate anywhere in the circuit. Together with the fact that $\{H, T, \CZ\}$ is a universal gate set, it follows that post-selected IQP circuits can decide $\PP$-complete languages.

Here, we reproduce \thmref{thm:IQP_hardness} using our criterion, \thmref{thm:maintwo}, together with the same gate set $\S_{\IQP}$ that Bremner et al. used.

\afterpage{\clearpage}
\input{hardness_dependency_diagram}

\afterpage{\clearpage}
\input{classification_dependency_diagram}

\begin{proof}[Proof of \thmref{thm:IQP_hardness}] By \thmref{thm:maintwo}, it suffices to find a finite number of post-selection gadgets over $\S_{\IQP}$ whose normalized actions generate a non-elementary, non-discrete, and strictly loxodromic subgroup of $\SL(2;\CC)$. To this end, consider the 1-to-1 post-selection gadget
$$
\a = \begin{quantikz}
        \lstick{} & \gate{H} & \gate{T} & \gate{H}  & \rstick{}
    \end{quantikz}
$$

\noindent as well as the two 2-to-1 post-selection gadgets
$$
\b = \begin{quantikz}
        \lstick{} & \gate{H} & \gate{T} & \ctrl{1} & \gate{H}  & \rstick{\bra{0}} \\
        \lstick{\ket{0}} & \gate{H} & \gate{T^4} & \ctrl{0} & \gate{H}  & \rstick{}
    \end{quantikz}
\quad
\text{and}
\quad
\c = \begin{quantikz}
        \lstick{} & \gate{H} & & \ctrl{1} & \gate{H}  & \rstick{} \\
        \lstick{\ket{0}} & \gate{H} & \gate{T} & \ctrl{0} & \gate{H}  & \rstick{\bra{0}}
    \end{quantikz}
$$

A straightforward calculation shows
that the normalized actions of $\a$ and $\b$ are the unitary $\SL(2;\CC)$ matrices
$$
A \defeq \tilde{\AA}(\a) = \frac{1}{2e^{i\pi/8}}
\begin{pmatrix} 
1 + e^{i\pi/4} & 1 - e^{i\pi/4} \\ 
1 - e^{i\pi/4} & 1 + e^{i\pi/4}
\end{pmatrix}
\quad
\text{and}
\quad
B \defeq \tilde{\AA}(\b) = \frac{1}{\sqrt{2}}
\begin{pmatrix} 
e^{i\pi/8} & -e^{i\pi/8} \\ 
e^{-i\pi/8} & e^{-i\pi/8}
\end{pmatrix},
$$
while the normalized action of $\c$ is the \emph{non}-unitary $\SL(2;\CC)$ matrix
$$
C \defeq \tilde{\AA}(\c) = \frac{1}{\sqrt{1-i}}
\begin{pmatrix} 
1 & e^{i\pi/4} \\ 
e^{i\pi/4} & 1
\end{pmatrix}.
$$ 
Crucially, inverse gadgets of $\a$, $\b$, and $\c$ are also realizable over $\S_{\IQP}$. For $\a$, the inverse is the 1-to-1 post-selection gadget
$$
\a^{-1} = \begin{quantikz}
        \lstick{} & \gate{H} & \gate{T^7} & \gate{H}  & \rstick{}
    \end{quantikz}
$$
while for $\b$ and $\c$, the inverses are the 2-to-1 post-selection gadgets
$$
\b^{-1} = \begin{quantikz}
        \lstick{} & \gate{H} & \ctrl{1} & \gate{T^4} & \gate{H}  & \rstick{\bra{0}} \\
        \lstick{\ket{0}} & \gate{H} & \ctrl{0} & \gate{T^7} & \gate{H}  & \rstick{}
    \end{quantikz}
\quad
\text{and}
\quad
\c^{-1} = \begin{quantikz}
        \lstick{} & \gate{H} & \ctrl{1} & \gate{T^4} & \gate{H}  & \rstick{} \\
        \lstick{\ket{0}} & \gate{H} & \ctrl{0} & \gate{T^3} & \gate{H}  & \rstick{\bra{0}}
    \end{quantikz}
$$
These are inverses of $\a$, $\b$, and $\c$ in the sense of \defref{def:inversegadget} because $\tilde{\AA}(\a^{-1}) = \tilde{\AA}(\a)^{-1} = A^{-1}$, $\tilde{\AA}(\b^{-1}) = \tilde{\AA}(\b)^{-1} = B^{-1}$, and $\tilde{\AA}(\c^{-1}) = -\tilde{\AA}(\c)^{-1} = -C^{-1}$.

Now let $\Gamma_{\IQP} \defeq \{A, A^{-1}, B, B^{-1}, C, -C^{-1}\}.$ Evidently, $\words{\Gamma_{\IQP}}$ is closed under inverses, so $\words{\Gamma_{\IQP}}$ is a subgroup of $\SL(2;\CC)$. In fact, $\words{\Gamma_\IQP}$ is a non-elementary, non-discrete, and strictly loxodromic subgroup of $\SL(2;\CC)$.
\begin{claim}
\label{claim:iqp_1}
$\words{\Gamma_{\IQP}}$ is a non-elementary subgroup of $\SL(2;\CC)$.
\end{claim}
\begin{proof}
Since $\beta(B) = -3 - \frac{1}{\sqrt{2}}$, $\beta(C) = -2 + 2i$, and $\gamma(B,C) = -1 + i$, $\IsElementary(\Gamma_\IQP) = \NO$.
\end{proof}

\begin{claim}
\label{claim:iqp_2}
$\words{\Gamma_{\IQP}}$ is a non-discrete subgroup of $\SL(2;\CC)$.
\end{claim}
\begin{proof}
Since $\tr(BAB^{-1}A^{-1}) = 1 + \frac{1}{\sqrt{2}} \neq 1$ and
$$
\left|\tr^2(B) - \tr(BAB^{-1}A^{-1})\right| + \left|\tr(BAB^{-1}A^{-1}) - 1\right| = \frac{1}{\sqrt{2}} < 1,
$$
is follows from line 16 in \algref{alg:discrete} that $\IsDiscrete(\Gamma_\IQP) = \NO$.
\end{proof}

\begin{claim}
\label{claim:iqp_3}
$\words{\Gamma_{\IQP}}$ is a strictly loxodromic subgroup of $\SL(2; \CC)$.
\end{claim}
\begin{proof}
Since $\tr(C) = \frac{2}{\sqrt{1 - i}} = \sqrt{2 + 2i} \in \mathbb{C} \backslash \RR$, $\IsLoxodromic(\Gamma_\IQP) = \YES$.
\end{proof}

Thus, the quantum advantage of IQP circuits follows from our criterion \thmref{thm:maintwo}.
\end{proof}

\subsection{Conjugated Clifford Circuits}
\label{sec:CCCs}

The Gottesman-Knill theorem proves that efficient classical computers can simulate uniform and polynomial size Clifford circuits \emph{exactly} \cite{Got98, AG04}. However, this result sensitively depends on an efficient state representation that is afforded by the Clifford group. It is therefore natural to wonder if by ``perturbing'' the Clifford group in some way, the perturbed circuits become hard to simulate classically.

CCCs are a type of perturbed Clifford circuit in which every $k$-qubit Clifford operation is conjugated by $U^{\otimes k}$ for some fixed single-qubit unitary $U$.
\begin{definition}
Fix $U \in \U(2)$. A \emph{$U$-conjugated Clifford circuit} (or \emph{$U$-CCC} for short) is an efficient quantum computer over the gate set\footnote{In the original CCC paper \cite{BFK18}, the authors use the gate set $\{U^\dagger H U, U^\dagger S U, (U^\dagger \otimes U^\dagger) \CNOT (U \otimes U)\}$. However, since $\CNOT = (I_1 \otimes H)\CZ(I_1 \otimes H)$, our gate set $\S_{\CCC}(U)$ is equivalent to theirs.}
$$
\S_{\CCC}(U) \defeq \big\{U^\dagger H U, U^\dagger S U, (U^\dagger \otimes U^\dagger) \CZ (U \otimes U) \big\}.
$$ 
A \emph{conjugated Clifford circuit} is a $U$-CCC for some $U \in \U(2)$.
\end{definition}

In \cite{BFK18}, Bouland, Fitzsimons, and Koh prove that when augmented with post-selection, CCCs can decide $\PP$-complete languages. Therefore, by an argument that is nearly identical to our \propref{prop:hardness}, one gets the following theorem.

\begin{theorem}[Corollary of Theorem 3.2 in \cite{BFK18}]
\label{thm:CCC_hardness}
If the polynomial hierarchy is infinite, then efficient classical computers cannot simulate CCCs to within multiplicative error $\epsilon < \sqrt{2} - 1$.
\end{theorem}

Bouland et al. prove Theorem \ref{thm:CCC_hardness} by employing the well-known result that if $V$ is any non-Clifford single-qubit gate, then $\{V, H, S, \CZ\}$ is a universal gate set \cite{NRS01}. They then show that if $U \neq e^{i\alpha} C R_z(\lambda)$ for all $C \in \words{H,S}$ and all $\alpha, \lambda \in [0,2\pi)$, then there exists a post-selection gadget over $\S_{\CCC}(U)$ that realizes a unitary non-Clifford gate $V$. Together, these facts imply that post-selected $U$-CCCs can decide $\PP$-complete languages. When combined with the Gottesman-Knill theorem, their techniques complete the complexity classification of $U$-CCCs:
\begin{restatable}[Theorem 3.2 in \cite{BFK18}]{theorem}{CCCclassification}
\label{thm:CCC_classification}
If the polynomial hierarchy is infinite, then efficient classical computers can simulate $U$-CCCs to within multiplicative error $\epsilon < \sqrt{2} - 1$ iff there exists $C \in \words{H,S}$ and $\alpha, \lambda \in [0,2\pi)$ such that $U = e^{i \alpha} C R_z(\lambda)$.
\end{restatable}

Here, we reproduce \thmref{thm:CCC_hardness} using our criterion. In \appendref{append:CCC_classification}, we reproduce the full complexity classification of CCC's (\thmref{thm:CCC_classification}), again using our criterion.

\begin{proof}[Proof of \thmref{thm:CCC_hardness}] It suffices to exhibit a single-qubit unitary $U$ such that a finite number of post-selection gadgets over $\S_{\CCC}(U)$ have normalized actions that 
generate a non-elementary, non-discrete, and strictly loxodromic subgroup of $\SL(2;\CC)$. To this end, let
$$
U = R_x\left(\frac{2\pi}{3}\right) = 
\begin{pmatrix}
\cos \frac{\pi}{3} & -i\sin \frac{\pi}{3}\\
-i\sin \frac{\pi}{3} & \cos \frac{\pi}{3}
\end{pmatrix}
$$
and consider the three 2-to-1 post-selection gadgets $\d$, $\e$, and $\f$ in \tabref{tab:ccc_gadgets}.

\begin{table}[htpb]
    \centering
    \begingroup
    \renewcommand*{\arraystretch}{1.5}
     \begin{tabular}{|c|c||c|c|}
        \hline
        \ & Gadget in $\S_{\CCC}(U)$ & \ & Inverse gadget in $\S_{\CCC}(U)$ 
        \\
        \hline\hline
$\d$ &
    \begin{quantikz}
    \\
        \lstick{} & \gate{U} & & \ctrl{1} &  \gate{U^\dagger}  & \rstick{} \\
        \lstick{\ket{0}} & \gate{U} & \gate{Z} & \ctrl{0} & \gate{U^\dagger}  & \rstick{\bra{0}}
    \\
    \end{quantikz}
        & 
$\d^{-1}$ &        
    \begin{quantikz}
    \\
        \lstick{} & \gate{U} & \gate{Z} & \ctrl{1} & \gate{U^\dagger}  & \rstick{} \\
        \lstick{\ket{1}} & \gate{U} & & \ctrl{0} & \gate{U^\dagger}  & \rstick{\bra{1}}
    \\
    \end{quantikz} \\ \hline
$\e$ & 
    \begin{quantikz}
    \\
        \lstick{} & \gate{U} & & \ctrl{1} & \gate{U^\dagger}  & \rstick{\bra{0}} \\
        \lstick{\ket{0}} & \gate{U} & \gate{Z} & \ctrl{0} & \gate{U^\dagger}  & \rstick{}
    \\
    \end{quantikz}
    & 
$\e^{-1}$ & 
    \begin{quantikz}
    \\
        \lstick{} & \gate{U} & & \ctrl{1} & \gate{U^\dagger}  & \rstick{\bra{1}} \\
        \lstick{\ket{1}} & \gate{U} & \gate{Z} & \ctrl{0} & \gate{U^\dagger}  & \rstick{}
    \\
    \end{quantikz} \\ \hline
$\f$ &
    \begin{quantikz}
    \\
        \lstick{} & \gate{U} & \gate{Z} & \ctrl{1} & \gate{U^\dagger}  & \rstick{} \\
        \lstick{\ket{0}} & \gate{U} & & \ctrl{0} & \gate{U^\dagger}  & \rstick{\bra{0}}
    \\
    \end{quantikz}  
    & 
$\f^{-1}$ &
    \begin{quantikz}
    \\
        \lstick{} & \gate{U} & & \ctrl{1} & \gate{U^\dagger}  & \rstick{} \\
        \lstick{\ket{1}} & \gate{U} & \gate{Z} & \ctrl{0} & \gate{U^\dagger}  & \rstick{\bra{1}}
    \\
    \end{quantikz} \\ \hline
\end{tabular}
\endgroup
\caption{Post-selection gadgets for $U$-CCCs with $U = R_x(\frac{2\pi}{3})$.}
\label{tab:ccc_gadgets}
\end{table}

A straightforward calculation shows
that the normalized actions of $\d$, $\e$, and $\f$ are, respectively, the non-unitary $\SL(2;\CC)$ matrices:
\begin{align*}
D &\defeq \tilde{\AA}(\d) = \frac{1}{4\sqrt{2}}
\begin{pmatrix} 
-5i & 3\sqrt{3} \\ 
-3\sqrt{3} & i
\end{pmatrix}\\
E &\defeq \tilde{\AA}(\e) = \frac{1}{2\sqrt{6}}
\begin{pmatrix} 
5 & 3i\sqrt{3} \\ 
i\sqrt{3} & 3
\end{pmatrix}\\
F &\defeq \tilde{\AA}(\f) = \frac{1}{4\sqrt{2}}
\begin{pmatrix} 
5 & -i\sqrt{3} \\ 
i\sqrt{3} & 7
\end{pmatrix}.
\end{align*}

Crucially, inverse gadgets of $\d$, $\e$, and $\f$ are also realizable over $\S_{\CCC}(U)$. They are $\d^{-1}$, $\e^{-1}$, and $\f^{-1}$ in \tabref{tab:ccc_gadgets}, respectively. Indeed, these are inverses of $\d$, $\e$, and $\f$ because $\tilde{\AA}(\d^{-1}) = \tilde{\AA}(\d)^{-1} = D^{-1}$, $\tilde{\AA}(\e^{-1}) = \tilde{\AA}(\e)^{-1} = E^{-1}$, and $\tilde{\AA}(\f^{-1}) = \tilde{\AA}(\f)^{-1} = F^{-1}$.

Now let $\Gamma_{\CCC}(U) \defeq \{D, D^{-1}, E, E^{-1}, F, F^{-1}\}.$ Evidently, $\words{\Gamma_{\CCC}(U)}$ is closed under inverses, so $\words{\Gamma_{\CCC}(U)}$ is a subgroup of $\SL(2;\CC)$. In fact, $\words{\Gamma_{\CCC}(U)}$ is a non-elementary, non-discrete, and strictly loxodromic subgroup of $\SL(2;\CC)$.
\begin{claim}
\label{claim:CCC_1}
$\words{\Gamma_{\CCC}(U)}$ is a non-elementary subgroup of $\SL(2;\CC)$.
\end{claim}
\begin{proof}
Since $\beta(E) = -\frac{4}{3}$, $\beta(F) = \frac{1}{2}$, and $\gamma(E,F) = \frac{1}{4}$, $\IsElementary(\Gamma_{\CCC}(U)) = \NO$.
\end{proof}

\begin{claim}
\label{claim:CCC_2}
$\words{\Gamma_{\CCC}(U)}$ is a non-discrete subgroup of $\SL(2;\CC)$.
\end{claim}
\begin{proof}
Since $\IsElementary(\{E,F\}) = \NO$ and
$$
\left|\tr^2(F) - 4\right| + \left|\tr(FEF^{-1}E^{-1}) - 2 \right| = \frac{3}{4} < 1,
$$
it follows from line 4 in \algref{alg:discrete} that $\IsDiscrete(\Gamma_{\CCC}(U)) = \NO$.
\end{proof}

\begin{claim}
\label{claim:CCC_3}
$\words{\Gamma_{\CCC}(U)}$ is a strictly loxodromic subgroup of $\SL(2; \CC)$.
\end{claim}
\begin{proof}
Since $\tr(D) = -\frac{i}{\sqrt{2}} \in \mathbb{C} \backslash \RR$, $\IsLoxodromic(\Gamma_{\CCC}(U)) = \YES$.
\end{proof}

Thus, the quantum advantage of CCCs follows from our criterion \thmref{thm:maintwo}. 
\end{proof}

\subsection{Conjugated $\CZ$ Circuits}

In this section, we prove a quantum advantage result for a special subclass of CCCs that we call \emph{conjugated $\CZ$ circuits}. As the name suggests, conjugated $\CZ$ circuits are CCCs where the interstitial Clifford circuit is made entirely of $\CZ$ gates. 

\begin{definition}
Fix $U \in \U(2)$. A \emph{$U$-conjugated $\CZ$ circuit} is an efficient quantum computer over the gate set
$$
\S_{\CZ}(U) \defeq \{(U^\dagger \otimes U^\dagger)\CZ(U \otimes U)\}.
$$
A \emph{conjugated $\CZ$ circuit} is a $U$-conjugated $\CZ$ circuit for some $U \in \U(2)$.
\end{definition}

Conjugated $\CZ$ circuits are incredibly simple. Nevertheless, it is very unlikely that efficient classical computers can simulate them in the weak multiplicative sense.

\begin{theorem}
\label{thm:CZhardness}
If the polynomial hierarchy is infinite, then efficient classical computers cannot simulate conjugated $\CZ$ circuits to within multiplicative error $\epsilon < \sqrt{2} - 1$.
\end{theorem}

To prove this, we employ an ostensibly less restrictive model of quantum computation.
\begin{definition}
Fix $U \in \U(2)$. A \emph{$U$-conjugated $\CZ + Z$ circuit} is an efficient quantum computer over the gate set
$$
\S_{\CZ + Z}(U) \defeq \{U^\dagger Z U, (U^\dagger \otimes U^\dagger)\CZ(U \otimes U)\}.
$$
A \emph{conjugated $\CZ + Z$ circuit} is a $U$-conjugated $\CZ + Z$ circuit for some $U \in \U(2)$.
\end{definition}

Conjugated $\CZ$ circuits are a type of conjugated $\CZ + Z$ circuit, so indeed conjugated $\CZ + Z$ circuits are less restrictive. In fact, we have already proven that efficient classical computers can most likely not simulate conjugated $\CZ + Z$ circuits in the weak multiplicative sense.

\begin{theorem}
\label{thm:czzhardnessresult}
If the polynomial hierarchy is infinite, then efficient classical computers cannot simulate conjugated $\CZ + Z$ circuits to within multiplicative error $\epsilon < \sqrt{2} - 1$.
\end{theorem}
\begin{proof}
Put $U = R_x(\frac{2\pi}{3})$ as in \sref{sec:CCCs}. Since the gadgets $\d, \e, \f, \d^{-1}, \e^{-1}$, and $\f^{-1}$ are all realizable over $\S_{\CZ + Z}(U)$, the result follows from our criterion.
\end{proof}

We now prove a lemma which together with \thmref{thm:czzhardnessresult} implies \thmref{thm:CZhardness}.

\begin{lemma}
\label{lem:CZlemma}
If the polynomial hierarchy is infinite, then for all $\theta \not\in \frac{\pi}{2}\ZZ$, efficient classical computers can simulate $R_x(\theta)$-conjugated $\CZ$ circuits to within multiplicative error $\epsilon < \sqrt{2} - 1$ iff they can simulate $R_x(\theta)$-conjugated $\CZ + Z$ circuits to within the same multiplicative error.
\end{lemma}
\begin{proof}
By \propref{prop:hardness}, it suffices to prove that for all $\theta \not\in \frac{\pi}{2}\ZZ$, $\GadBQP(\S_{\CZ}(R_x(\theta))) = \GadBQP(\S_{\CZ + Z}(R_x(\theta)))$. It is plain that $\GadBQP(\S_{\CZ}(R_x(\theta))) \subseteq \GadBQP(\S_{\CZ + Z}(R_x(\theta)))$ because every $R_x(\theta)$-conjugated $\CZ$ circuit is an $R_x(\theta)$-conjugated $\CZ + Z$ circuit. For the other direction, consider the following gadget over $\S_{\CZ}(R_x(\theta))$:
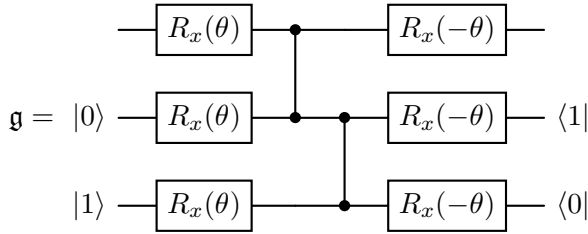
$$
\g =
\begin{quantikz}
    \\
        \lstick{} & \gate{R_x(\theta)} & \ctrl{1} &  & \gate{R_x(-\theta)} & \rstick{} \\
        \lstick{\ket{0}} & \gate{R_x(\theta)} & \ctrl{0} & \ctrl{1} & \gate{R_x(-\theta)} & \rstick{\bra{1}} \\
        \lstick{\ket{1}} & \gate{R_x(\theta)} &  & \ctrl{0} & \gate{R_x(-\theta)} & \rstick{\bra{0}} 
    \\
    \end{quantikz}
$$
It is straightforward to show that $\det \AA(\g) = -\frac{1}{4}\sin^4(\theta)$. Therefore, the normalized action $\tilde{\AA}(\g)$ exists if $\theta \not\in\frac{\pi}{2}\ZZ$. In this case,
$$
\tilde{\AA}(\g) = 
\begin{pmatrix}
i\cos(\theta) & \sin(\theta)\\
-\sin(\theta) & -i\cos(\theta)
\end{pmatrix}
= iR_x(\theta)ZR_x(-\theta).
$$
Consequently, if $\theta \not\in\frac{\pi}{2}\ZZ$, then there exists a gadget over $\S_{\CZ}(R_x(\theta))$ that exactly implements an $R_x(\theta)$-conjugated $Z$ gate (up to the immaterial phase $i$), so every efficient gadget quantum computer over $S_{\CZ + Z}(R_x(\theta))$ can be exactly simulated by an efficient gadget quantum computer over $S_{\CZ}(R_x(\theta))$. This implies $\GadBQP(\S_{\CZ + Z}(R_x(\theta)))\subseteq \GadBQP(\S_{\CZ}(R_x(\theta)))$, as desired.
\end{proof}

We now prove \thmref{thm:CZhardness}.

\begin{proof}[Proof of \thmref{thm:CZhardness}]
Put $U = R_x(\frac{2\pi}{3})$ and suppose that the polynomial hierarchy is infinite. Then, by the proof of \thmref{thm:czzhardnessresult}, no efficient classical computer can simulate $U$-conjugated $\CZ + Z$ circuits to within multiplicative error $\epsilon < \sqrt{2} - 1$. Therefore, by \lemref{lem:CZlemma} and the fact that $\frac{2\pi}{3} \not\in \frac{\pi}{2}\ZZ$, no efficient classical computer can simulate $U$-conjugated $\CZ$ circuits to within the same multiplicative error.
\end{proof}

In \appendref{append:czclassification}, we prove the full complexity classification of $U$-conjugated $\CZ$ circuits.

\begin{restatable}{theorem}{CZclassification}
\label{thm:czclassification}
If the polynomial hierarchy is infinite, then efficient classical computers can simulate $U$-conjugated $\CZ$ circuits to within multiplicative error $\epsilon < \sqrt{2} - 1$ iff there exists $C \in \words{H,S}$ and $\alpha, \phi, \lambda \in [0,2\pi)$ such that $U = e^{i\alpha} R_z(\phi)CR_z(\lambda)$.
\end{restatable}

Interestingly, in Appendices~\ref{append:czzclassification} and \ref{append:czsclassification}, respectively, we are also able to prove that this same classification applies to conjugated $\CZ + Z$ circuits and so-called \emph{conjugated $\CZ + S$ circuits}, which are defined exactly like conjugated $\CZ + Z$ circuits but with $Z$ replaced by $S$.

\begin{restatable}{theorem}{CZZclassification}
\label{thm:cz+zclassification}
If the polynomial hierarchy is infinite, then efficient classical computers can simulate $U$-conjugated $\CZ + Z$ circuits to within multiplicative error $\epsilon < \sqrt{2} - 1$ iff there exists $C \in \words{H,S}$ and $\alpha, \phi, \lambda \in [0,2\pi)$ such that $U = e^{i\alpha} R_z(\phi)CR_z(\lambda)$.
\end{restatable}

\begin{restatable}{theorem}{CZSclassification}
\label{thm:cz+sclassification}
If the polynomial hierarchy is infinite, then efficient classical computers can simulate $U$-conjugated $\CZ + S$ circuits to within multiplicative error $\epsilon < \sqrt{2} - 1$ iff there exists $C \in \words{H,S}$ and $\alpha, \phi, \lambda \in [0,2\pi)$ such that $U = e^{i\alpha} R_z(\phi)CR_z(\lambda)$.
\end{restatable}

We remark that conjugated $\CZ + S$ circuits are also known as \emph{commuting conjugated Clifford circuits}. Commuting CCCs are a sort of intersection of IQP circuits and CCCs. Given the quantum advantage of both IQP circuits and CCCs, Bouland et al. ask in \cite{BFK18} if commuting CCCs afford a quantum advantage. Indeed, our \thmref{thm:cz+sclassification} not only resolves this question, but it also gives the full complexity classification of commuting CCCs.

In the next section, we study the complexity of CCCs over other subsets of the Clifford group.

\afterpage{\clearpage}
\input{multi_qubit_lattice}

\subsection{Conjugated Clifford Fragment Circuits}


In 1941, Emil Post published a complete classification of all the ways in which sets of Boolean logic gates can fail to be universal \cite{Pos41}. In \cite{AGS17}, Aaronson, Grier, and Schaeffer describe the ambitious program of doing the same but for all quantum gates. In \cite{GS22}, Grier and Schaeffer completed such a classification for all Clifford gates. They found that any collection of Clifford gates generates one of 57 distinct ``classes'' or ``fragments'' of the Clifford group.

In \cite{BFK18}, Bouland et al. ask if efficient classical computers can simulate CCCs when the interstitial Clifford circuit is restricted to one of these fragments. We call these \emph{conjugated Clifford fragment circuits}, or CCFCs for short. Of course, while CCFCs derive from the specific Grier-Schaeffer classification of the Clifford group---a classification that does not obviously also classify the Clifford group when it is conjugated by some $U \in \U(2)$---this classification nevertheless affords a myriad of interesting restricted models (the various CCFCs) that we can use to demonstrate the utility of our criterion\footnote{Specifically, the Grier-Schaeffer classification allows for Clifford circuits to make use of arbitrary ancillary workspace qubits, as long as they are completely uncorrelated with the input state at both the beginning and end of the circuit. This resource does not obviously extend naturally to the conjugated model we consider, and so we disregard this ``ancilla rule", focusing exclusively on the generators of the various fragments.}.

To this end, we first note that of the 57 Clifford fragments in the Grier-Schaeffer classification, there are 30 that are generated by single-qubit gates alone. These are obviously insufficient to afford any sort of quantum advantage because they contain no entangling gates \cite{JL03}. This leaves $57 - 30 = 27$ fragments for which the corresponding CCFC may afford some sort of quantum advantage. The inclusion lattice for the generators of these 27 fragments is shown in \figref{fig:Cliffordlattice}, which is from \cite{GS22}.

In this section, we use our criterion to prove that if the polynomial hierarchy is infinite, then no efficient classical computer can simulate CCFCs to within small multiplicative error for all 27 of the remaining fragments\footnote{It is worth noting that this does \emph{not} mean that CCCs over any multi-qubit subgroup of the Clifford group are classically intractable, as under the ancilla rule several otherwise weak subgroups might collapse into a single fragment in the classification.}. This resolves the question raised by Bouland et al. \cite{BFK18}.

We now define the generators of some of the fragments in \figref{fig:Cliffordlattice}.

The set $\Pauli$ is the single-qubit Pauli group, i.e., the group generated by the three single-qubit Pauli gates:
$$
X = 
\begin{pmatrix}
0 & 1\\
1 & 0
\end{pmatrix},
\quad Y = 
\begin{pmatrix}
0 & -i\\
i & 0
\end{pmatrix},
\quad
\text{and}
\quad
Z = 
\begin{pmatrix}
1 & 0\\
0 & -1
\end{pmatrix}.
$$
Given $P, Q \in \{X,Y,Z\}$, we define a generalized $\CNOT$ gate by
$$
C(P,Q) \defeq \frac{I_1 \otimes I_1 + P \otimes I_1 + I_1 \otimes Q - P \otimes Q}{2},
$$
where if the first qubit is in the $+1$ eigenspace of $P$, then $C(P, Q)$
does nothing, but if it is in the $-1$ eigenspace of $P$, then $C(P, Q)$ applies $Q$ to the second qubit. Therefore, $C(Z,Z)$ is $\CZ$. 

\noindent The $R_Z$ gate is the $S$ gate (i.e., a $\pi/2$ rotation about the $\hat{z}$-axis of the Bloch sphere); the $R_X$ and $R_Y$ gates are the same rotation but about the $\hat{x}$- and $\hat{y}$-axes, respectively:
$$
R_X = \frac{I_1 - iX}{\sqrt{2}} \quad \text{and} \quad R_Y = \frac{I_1 - iY}{\sqrt{2}}.
$$

The $\theta_{X+Z}$ gate is the Hadamard gate (i.e., a $\pi$ rotation about the $(\hat{x} + \hat{z}) / \sqrt{2}$ axis of the Bloch sphere); the $\theta_{Y+Z}$ and $\theta_{X+Y}$ gates are the same rotation but about the $(\hat{y} + \hat{z}) / \sqrt{2}$- and $(\hat{x} + \hat{y}) / \sqrt{2}$-axes, respectively:
$$
\theta_{Y+Z} = \frac{Y+Z}{\sqrt{2}} \quad \text{and} \quad \theta_{X+Y} = \frac{X + Y}{\sqrt{2}}.
$$ 
For $P,Q \in \{X,Y,Z\}$, one can analogously define a $\theta_{P-Q}$ gate. In \figref{fig:Cliffordlattice}, fragments with $\theta_{PQ}$ are fragments that contain both $\theta_{P+Q}$ and $\theta_{P-Q}$.

Finally, for all $k \in \NN$, the $\T_{2k}$ gate is such that for every $x = (x_1, x_2, \dots, x_{2k}) \in \{0,1\}^{2k}$,
$$
\T_{2k}\ket{x_1, x_2, \dots, x_{2k}} = \ket{x_1 \oplus b_x, x_2 \oplus b_x, \dots, x_{2k} \oplus b_x},
$$
where $\oplus$ is addition modulo 2 and $b_x = x_1 \oplus x_2 \oplus \cdots \oplus x_{2k}$. In other words, for computational basis inputs, $\T_{2k}$ outputs the complement of the input when the parity of the input is odd and does nothing when the parity of the input is even. Therefore, $\T_2$ is the SWAP gate, and $\T_4$ maps $\ket{1011} \mapsto \ket{0100}$ and $\ket{1100} \mapsto \ket{1100}$.

We now formally define the CCFC model. In the following, $\deg(P) = k$ iff $P \in \U(2^k)$.

\begin{definition}
Fix $U \in \U(2)$ and let $\mathcal{F}$ be the generators of one of the 27 fragments depicted in \figref{fig:Cliffordlattice}. A \emph{$U$-conjugated $\mathcal{F}$ circuit} is an efficient quantum computer over the gate set
$$
\S_{\mathcal{F}}(U) \defeq \big\{ (U^\dagger)^{\otimes \deg(P)} P U^{\otimes \deg(P)} \mid P \in \mathcal{F}\big\}.
$$
A \emph{conjugated $\mathcal{F}$ circuit} is a $U$-conjugated $\mathcal{F}$ circuit for some $U \in \U(2)$.
\end{definition}

In this section, we prove the following quantum advantage result using our criterion. 

\begin{theorem}
\label{thm:hardness_fragmenttwo}
If the polynomial hierarchy is infinite, then for all fragments $\mathcal{F}$ in \figref{fig:Cliffordlattice}, efficient classical computers cannot simulate conjugated $\mathcal{F}$ circuits to within multiplicative error $\epsilon < \sqrt{2} - 1$.
\end{theorem}

A main ingredient of the proof of \thmref{thm:hardness_fragmenttwo} is the following hardness result for CCFCs over the $\T_4 + \Pauli$ fragment, which is a special case of \thmref{thm:hardness_fragmenttwo}.

\begin{theorem}
\label{thm:t4hardness}
If the polynomial hierarchy is infinite, then efficient classical computers cannot simulate conjugated $\T_4 + \Pauli$ circuits to within multiplicative error $\epsilon < \sqrt{2} - 1$.
\end{theorem}

We also establish the full complexity classification for the $C(P,P) + R_P$ and $C(P,P) + P$ fragments:
\begin{theorem}
\label{thm:hardness_fragmentone}
If the polynomial hierarchy is infinite, then for all $P \in \{X,Y,Z\}$, efficient classical computers can simulate $U$-conjugated $C(P,P) + R_P$ circuits to within multiplicative error $\epsilon < \sqrt{2} - 1$ iff there exists $C \in \words{H,S}$ and $\alpha, \phi, \lambda \in [0,2\pi)$ such that:
\begin{enumerate}[(i)]
\item if $P = X$, then $U = e^{i\alpha}HR_z(\phi)CR_z(\lambda)$,
\item if $P = Y$, then $U = e^{i\alpha}\theta_{Y + Z}R_z(\phi)CR_z(\lambda)$,
\item if $P = Z$, then $U = e^{i\alpha}R_z(\phi)CR_z(\lambda)$.
\end{enumerate}
\end{theorem}

\begin{theorem}
\label{thm:hardness_fragmentthree}
If the polynomial hierarchy is infinite, then for all $P \in \{X,Y,Z\}$, efficient classical computers can simulate $U$-conjugated $C(P,P) + P$ circuits to within multiplicative error $\epsilon < \sqrt{2} - 1$ iff there exists $C \in \words{H,S}$ and $\alpha, \phi, \lambda \in [0,2\pi)$ such that:
\begin{enumerate}[(i)]
\item if $P = X$, then $U = e^{i\alpha}HR_z(\phi)CR_z(\lambda)$,
\item if $P = Y$, then $U = e^{i\alpha}\theta_{Y + Z}R_z(\phi)CR_z(\lambda)$,
\item if $P = Z$, then $U = e^{i\alpha}R_z(\phi)CR_z(\lambda)$.
\end{enumerate}
\end{theorem}

In fact, in consequence to our complexity classification for conjugated $\CZ$ circuits (\thmref{thm:czclassification}), we can easily obtain a complexity classification for all conjugated $C(P,P)$ circuits, for any $P \in \{X,Y,Z\}$.

\begin{theorem}
\label{thm:hardness_fragmentfour}
If the polynomial hierarchy is infinite, then for all $P \in \{X,Y,Z\}$, efficient classical computers can simulate $U$-conjugated $C(P,P)$ circuits to within multiplicative error $\epsilon < \sqrt{2} - 1$ iff there exists $C \in \words{H,S}$ and $\alpha, \phi, \lambda \in [0,2\pi)$ such that:
\begin{enumerate}[(i)]
\item if $P = X$, then $U = e^{i\alpha}HR_z(\phi)CR_z(\lambda)$,
\item if $P = Y$, then $U = e^{i\alpha}\theta_{Y + Z}R_z(\phi)CR_z(\lambda)$,
\item if $P = Z$, then $U = e^{i\alpha}R_z(\phi)CR_z(\lambda)$.
\end{enumerate}
\end{theorem}

Since they are simpler, we begin by proving Theorems~\ref{thm:hardness_fragmentone} -- \ref{thm:hardness_fragmentfour}.

\begin{proof}[Proof of \thmref{thm:hardness_fragmentone}]
Conjugated $C(Z,Z) + R_Z$ circuits are exactly the commuting CCC (i.e. conjugated $\CZ + S$ circuit) model. Thus, this case follows from \thmref{thm:cz+sclassification}. For the remaining two fragments $C(X,X) + R_X$ and $C(Y,Y) + R_Y$, simply note that for all $U \in \U(2)$,
$$
\S_{C(X,X) + R_X}(HU) = \S_{C(Z,Z) + R_Z}(U) \quad \text{and} \quad \S_{C(Y,Y) + R_Y}(\theta_{Y + Z}U) = \S_{C(Z,Z) + R_Z}(U).
$$
Thus, the remaining two cases also follow from \thmref{thm:cz+sclassification}.
\end{proof}

The proofs of Theorems~\ref{thm:hardness_fragmentthree} and \ref{thm:hardness_fragmentfour} are analogous to the proof of \thmref{thm:hardness_fragmentone}, but instead rely on Theorems~\ref{thm:cz+zclassification} and \ref{thm:czclassification}, respectively.

We now prove that it is very likely that CCFCs over any fragment of the Clifford group afford a quantum advantage.

\begin{proof}[Proof of \thmref{thm:hardness_fragmenttwo}]
Assuming the polynomial hierarchy is infinite, it suffices to prove that efficient classical computers can neither simulate conjugated $\T_4 + \Pauli$ circuits to within multiplicative error $\epsilon < \sqrt{2} - 1$ nor can they simulate conjugated $C(P,P) + P$ circuits for any $P \in \{X,Y,Z\}$ to within the same multiplicative error. This suffices, because every fragment in \figref{fig:Cliffordlattice} contains either the $\T_4 + \Pauli$ fragment or one of the $C(P,P) + P$ fragments. The quantum advantage of the three $C(P,P) + P$ fragments follows from the classification result \thmref{thm:hardness_fragmentthree}, and the quantum advantage of the $\T_4 + \Pauli$ fragment follows from \thmref{thm:t4hardness}.
\end{proof}

Thus it remains to prove \thmref{thm:t4hardness}.

\begin{proof}[Proof of \thmref{thm:t4hardness}]
It suffices to exhibit a single-qubit unitary $U$ such that a finite number of post-selection gadgets over $\S_{\TP}(U)$ have normalized actions that generate a non-elementary, non-discrete, and strictly loxodromic subgroup of $\SL(2;\CC)$. To this end, let
$$
U = T R_x\left(\frac{2\pi}{3}\right) = 
\frac{1}{2}
\begin{pmatrix}
1 & -\sqrt{3}i\\
\sqrt{\frac{3}{2}}(1 - i) & e^{i\pi/4}
\end{pmatrix}
$$
and consider the three 4-to-1 post-selection gadgets $\h$, $\ii$, and $\j$ in \tabref{tab:t4+p_gadgets}.

\begin{table}[htpb]
    \centering
    \begingroup
    \renewcommand*{\arraystretch}{1.5}
     \begin{tabular}{|c|c||c|c|}
        \hline
        \ & Gadget in $\S_{\T_4 + \Pauli}(U)$ & \ & Inverse gadget in $\S_{\T_4 + \Pauli}(U)$ 
        \\
        \hline\hline
$\h$ &
    \begin{quantikz}
    \\
        \lstick{\ket{0}} & \gate{U} & \gate[4]{\T_4} & \gate{Z} & \gate{U^\dagger} & \rstick{} \\
        \lstick{\ket{0}} & \gate{U} & & \gate{Y} & \gate{U^\dagger} & \rstick{\bra{1}} \\
        \lstick{} & \gate{U} & & \gate{Y} & \gate{U^\dagger} & \rstick{\bra{0}} \\
        \lstick{\ket{0}} & \gate{U} & & \gate{X} & \gate{U^\dagger} & \rstick{\bra{1}}
    \\
    \end{quantikz}
        & 
$\h^{-1}$ &        
    \begin{quantikz}
    \\
        \lstick{} & \gate{U} & \gate{Z} & \gate[4]{\T_4} & \gate{U^\dagger} & \rstick{\bra{1}} \\
        \lstick{\ket{0}} & \gate{U} & \gate{Y} & & \gate{U^\dagger} & \rstick{} \\
        \lstick{\ket{0}} & \gate{U} & \gate{Y} & & \gate{U^\dagger} & \rstick{\bra{1}} \\
        \lstick{\ket{0}} & \gate{U} & \gate{X} & & \gate{U^\dagger} & \rstick{\bra{0}}
    \\
    \end{quantikz} \\ \hline
$\ii$ & 
    \begin{quantikz}
    \\
        \lstick{} & \gate{U} & \gate[4]{\T_4} & \gate{Z} & \gate{U^\dagger} & \rstick{} \\
        \lstick{\ket{1}} & \gate{U} & & \gate{X} & \gate{U^\dagger} & \rstick{\bra{0}} \\
        \lstick{\ket{0}} & \gate{U} & & \gate{Z} & \gate{U^\dagger} & \rstick{\bra{0}} \\
        \lstick{\ket{0}} & \gate{U} & & \gate{X} & \gate{U^\dagger} & \rstick{\bra{0}}
    \\
    \end{quantikz}
    & 
$\ii^{-1}$ & 
    \begin{quantikz}
    \\
        \lstick{} & \gate{U} & \gate{X} & \gate[4]{\T_4} & \gate{U^\dagger} & \rstick{} \\
        \lstick{\ket{1}} & \gate{U} & \gate{Z} & & \gate{U^\dagger} & \rstick{\bra{0}} \\
        \lstick{\ket{0}} & \gate{U} & \gate{X} & & \gate{U^\dagger} & \rstick{\bra{0}} \\
        \lstick{\ket{0}} & \gate{U} & \gate{Z} & & \gate{U^\dagger} & \rstick{\bra{0}}
    \\
    \end{quantikz} \\ \hline
$\j$ &
    \begin{quantikz}
        \lstick{} & \gate{U} & \gate[4]{\T_4} & \gate{Y} & \gate{U^\dagger} & \rstick{} \\
        \lstick{\ket{0}} & \gate{U} & & \gate{X} & \gate{U^\dagger} & \rstick{\bra{0}} \\
        \lstick{\ket{0}} & \gate{U} & & \gate{Y} & \gate{U^\dagger} & \rstick{\bra{0}} \\
        \lstick{\ket{0}} & \gate{U} & & \gate{X} & \gate{U^\dagger} & \rstick{\bra{0}}
    \end{quantikz}  
    & 
$\j^{-1}$ &
    \begin{quantikz}
    \\
        \lstick{} & \gate{U} & \gate{Y} & \gate[4]{\T_4} & \gate{U^\dagger} & \rstick{} \\
        \lstick{\ket{1}} & \gate{U} & \gate{X} & & \gate{U^\dagger} & \rstick{\bra{1}} \\
        \lstick{\ket{1}} & \gate{U} & \gate{X} & & \gate{U^\dagger} & \rstick{\bra{1}} \\
        \lstick{\ket{1}} & \gate{U} & \gate{Y} & & \gate{U^\dagger} & \rstick{\bra{1}}
    \\
    \end{quantikz} \\ \hline
\end{tabular}
\endgroup
\caption{Post-selection gadgets for $U$-conjugated $\T_4 + \Pauli$ Clifford circuits with $U = TR_x(\frac{2\pi}{3})$.}
\label{tab:t4+p_gadgets}
\end{table}

The normalized action of $\h$ is the unitary $\SL(2;\CC)$ matrix
$$
H \defeq \tilde{\AA}(\h) = 
\frac{1}{2\sqrt{2}}
\begin{pmatrix}
2 - i & \sqrt{3}\\
-\sqrt{3} & 2 + i
\end{pmatrix},
$$
while the normalized actions of $\ii$ and $\j$ are the \emph{non}-unitary $\SL(2;\CC)$ matrices
$$
I \defeq \tilde{\AA}(\ii) = 
\frac{1}{5}
\begin{pmatrix}
\sqrt{5}(2 - i) & 0\\
-2 - 4i & \sqrt{3}(2 + i)
\end{pmatrix}
\quad
\text{and}
\quad
J \defeq \tilde{\AA}(\j) = 
\frac{1}{10}
\begin{pmatrix}
-3\sqrt{3}i & -11\\
11 & -\frac{7i}{\sqrt{3}}
\end{pmatrix}.
$$

Crucially, inverse gadgets of $\h$, $\ii$, and $\j$ are also realizable over $\S_{\TP}(U)$. They are $\h^{-1}$, $\ii^{-1}$, and $\j^{-1}$ in \tabref{tab:t4+p_gadgets}, respectively. These are inverses of $\h$, $\ii$, and $\j$ because $\tilde{\AA}(\h^{-1}) = -\tilde{\AA}(\h)^{-1} = -H^{-1}$, $\tilde{\AA}(\ii^{-1}) = \tilde{\AA}(\ii)^{-1} = I^{-1}$, and $\tilde{\AA}(\j^{-1}) = \tilde{\AA}(\j)^{-1} = J^{-1}$.

Now let $\Gamma_{\TP}(U) \defeq \{H, -H^{-1}, I, I^{-1}, J, J^{-1}\}.$ Evidently, $\words{\Gamma_{\TP}(U)}$ is closed under inverses, so $\words{\Gamma_{\TP}(U)}$ is a subgroup of $\SL(2;\CC)$. In fact, $\words{\Gamma_{\TP}(U)}$ is a non-elementary, non-discrete, and strictly loxodromic subgroup of $\SL(2;\CC)$.
\begin{claim}
\label{claim:T4_1}
$\words{\Gamma_{\TP}(U)}$ is a non-elementary subgroup of $\SL(2;\CC)$.
\end{claim}
\begin{proof}
Since $\beta(H) = -2$, $\beta(I) = -\frac{4}{5}$, and $\gamma(H,I) = -\frac{36}{125} + \frac{48i}{125}$, $\IsElementary(\Gamma_{\TP}(U)) = \NO$.
\end{proof}

\begin{claim}
\label{claim:T4_2}
$\words{\Gamma_{\TP}(U)}$ is a non-discrete subgroup of $\SL(2;\CC)$.
\end{claim}
\begin{proof}
Since $\tr(HIH^{-1}I^{-1}) = \frac{214}{125} + \frac{48i}{125} \neq 1$ and
$$
\left|\tr^2(H) - 2\right| + \left|\tr(HIH^{-1}I^{-1}) - 1\right| = \frac{\sqrt{409}}{25} \approx 0.809 < 1,
$$
it follows from line 6 in \algref{alg:discrete} that $\IsDiscrete(\Gamma_{\TP}(U)) = \NO$.
\end{proof}

\begin{claim}
\label{claim:T4_3}
$\words{\Gamma_{\TP}(U)}$ is a strictly loxodromic subgroup of $\SL(2; \CC)$.
\end{claim}
\begin{proof}
Since $\tr(J) = -\frac{8i}{5\sqrt{3}} \in \mathbb{C} \backslash \RR$, $\IsLoxodromic(\Gamma_{\TP}(U)) = \YES$.
\end{proof}
Thus, the quantum advantage of conjugated $\T_4 + \mathcal{P}$ circuits follows from our criterion \thmref{thm:maintwo}.
\end{proof}

A natural problem is to complete the complexity classification of CCFCs for every fragment, such as conjugated $\T_4 + \Pauli$ circuits.

\section{Discussion}
\label{sec:discussion}

In this work, we have established a criterion (\thmref{thm:maintwo}) for testing if an efficient quantum computer over a non-universal gate set $\S$ can perform a sampling task that no efficient classical computer can, assuming standard complexity assumptions. In particular, we showed that if there exist a finite number of post-selection gadgets over $\S$ whose normalized actions generate a non-elementary, non-discrete, and strictly loxodromic subgroup of $\SL(2;\CC)$, then efficient classical computers cannot simulate efficient quantum computers over $\S$, assuming the polynomial hierarchy is infinite. We also demonstrated that this criterion is ``simple'' in the sense that there is a straightforward algorithm for checking if the normalized gadget actions generate a non-elementary, non-discrete, and strictly loxodromic subgroup of $\SL(2;\CC)$. That said, we acknowledge that there is no obvious, systematic way to find the gadgets to use in our criterion. On the other hand, our approach is simple enough that for ``small'' $j$, one can conceivably automate the search for the $j$-to-1 post-selection gadgets to use in our criterion.

Using our criterion, we reproduced the well-known quantum advantage results that IQP circuits and CCCs can perform a sampling task that no efficient classical computer can, assuming the polynomial hierarchy is infinite \cite{BJS11, BFK18}. We also proved that commuting CCCs and CCCs over every multi-qubit fragment of the Cliffford group (as classified in \cite{GS22}) can similarly perform a sampling task that no efficient classical computer can, again assuming the polynomial hierarchy is infinite. For CCCs, CCCs over $C(P,P)$, and CCCs over the $C(P,P) + P$ and $C(P,P) + R_P$ fragments, we proved the full complexity classification (Appendices~\ref{append:CCC_classification}, \ref{append:czclassification}, \ref{append:czzclassification}, and \ref{append:czsclassification}, respectively).

Our results for commuting CCCs and CCCs over the various fragments of the Clifford group resolve two open questions that were raised by Bouland et al. in \cite{BFK18}. We attribute our ability to address these questions to the historical fact that the ``standard" approach for showing classical intractability in the weak multiplicative sense is via \emph{unitary} gadget injection. On the other hand, the quantum advantage of commuting CCCs and CCCs over the various Clifford fragments follow from the group theory of $\SL(2;\CC)$, and hence via \emph{non}-unitary gadget injection. Incidentally, non-unitary gadget injection was also used to prove a quantum advantage result for two-qubit commuting Hamiltonians \cite{BMZ16}.

Our results lead to many natural open questions. First, throughout this paper we have assumed that the set of normalized gadget actions $\Gamma$ contains only the actions of normalizable $j$-to-$1$ gadgets for any $j$. The reason, of course, is because such gadgets have actions in $\SL(2;\CC)$. Consequently, the group theory of $\SL(2;\CC)$ can be used to infer their effect on quantum computation. However, at least on the surface, we see no substantive reason why more general $j$-to-$k$ gadgets cannot be used, save the epistemological fact that the group theory of $\SL(2^k; \CC)$ is not nearly as understood as the group theory of $\SL(2;\CC)$. In this regard, we wonder if there are statements similar to our criterion, but for more general $j$-to-$k$ post-selection gadgets.

A separate question has to do with the existence of inverse gadgets. As it stands, our criterion only works if there is an inverse of every gadget used. Of course, this is necessary for $\words{\Gamma}$ to be a subgroup of $\SL(2;\CC)$. In \sref{sec:applications}, we guarantee this by explicitly finding an inverse of each gadget. However, this approach does not easily scale. A means of improving our results, therefore, is to show that every $j$-to-$k$ post-selection gadget has an inverse. This way, any finite subset $\Gamma \subset \gad_k(\S)$ implies another finite subset $\Gamma' \subset \gad_k(\S)$ such that $\Gamma \subseteq \Gamma'$ and $\words{\Gamma'}$ is a subgroup of $\SL(2;\CC)$. If this is right, then one could assume without loss of generality that $\words{\Gamma}$ is always closed under inverses and hence that $\words{\Gamma}$ is always a subgroup of $\SL(2;\CC)$. Indeed, we conjecture this to be true:
\begin{conjecture}[Existence of Inverse Gadgets]
\label{conj:gadgetconj}
Let $\S$ be a gate set. For all $j$-to-$k$ post-selection gadgets $\g = \g_{j,k}$ over $\S$ for which $\det \AA(\g) \neq 0$, there exists a $j'$-to-$k$ post-selection gadget $\g^{-1} = \g_{j', k}$ over $\S$ such that $\AA(\g)^{-1} \in \words{\AA(\g), \AA(\g^{-1})}$.\footnote{Perhaps this is only true when $\S$ itself is closed under inverses.}
\end{conjecture}

Interestingly, if, for some gate set $\S$, $\Gamma \subset \gad_1(\S)$ is not closed under inverses but is nevertheless able to approximate every gate in $\Gamma^{-1} \defeq \{\omega^{-1} \mid \omega \in \Gamma\}$ arbitrarily well, then $\Gamma$ densely generates $\SL(2;\CC)$ if and only if $\words{\Gamma \cup \Gamma^{-1}}$ is a dense subgroup of $\SL(2;\CC)$. This implies the following slight refinement of our criterion:

\begin{theorem}
Let $\S$ be a gate set and suppose the polynomial hierarchy is infinite. If there exists a finite subset $\Gamma \subset \gad_1(\S)$ such that:
\begin{enumerate}[(i)]
\item for all $\delta > 0$ and all $\omega \in \Gamma$, there exists a sequence $\sigma_{\omega^{-1}}$ of gates from $\Gamma$ such that $\opnorm{\omega^{-1} - \sigma_{\omega^{-1}}} < \delta$,
\item $\words{\Gamma \cup \Gamma^{-1}}$ is a non-elementary, non-discrete, and strictly loxodromic subgroup of $\SL(2;\CC)$,
\end{enumerate}
then efficient classical computers cannot simulate efficient quantum computers over $\S$ to within multiplicative error $\epsilon < \sqrt{2} - 1$.
\label{thm:inversefreecriterion}
\end{theorem}

Another question concerns the need for strict loxodromy in Theorems~\ref{thm:maintwo} and \ref{thm:inversefreecriterion}. Recall that demanding that $\Gamma \subset \gad_1(\S)$ generate a strictly loxodromic subgroup $H = \words{\Gamma}$ of $\SL(2;\CC)$ forced $H$ to be dense in $\SL(2;\CC)$ by Sullivan's \thmref{thm:SullivanTheorem}. This, in turn, implies $\GadBQP(\S) = \PostBQP$. Technically, however, $H$ need not be dense in $\SL(2;\CC)$ for $\GadBQP(\S) = \PostBQP$. This holds, for example, if $H$ is dense in $\SU(2)$ or if $H$ is dense in $\SL(2;\RR)$ and the entangling gate in $\S$ is a real matrix \cite{BV93, McK13} (in which case one must also appeal to an $\SL(2;\RR)$ analogue of the Solovay-Kitaev theorem \cite{AAEL08}). With this intuition, we suspect that as long as $H$ is conjugate to a dense subgroup of $\SL(2;\RR)$, then $\GadBQP(\S) = \PostBQP$.

\begin{conjecture}[Irrelevance of Strict Loxodromy]
\label{conj:densityresult}
Let $\S$ be a gate set. If there exists a finite subset $\Gamma \subset \gad_1(\S)$ such that $\words{\Gamma}$ is conjugate to a dense subgroup of $\SL(2;\RR)$, then $\GadBQP(\S) = \PostBQP$.
\end{conjecture}

If this conjecture is right, then we extirpate the need for strict loxodromy in Theorems~\ref{thm:maintwo} and \ref{thm:inversefreecriterion}, and thereby obtain a stronger criterion for post-selected quantum advantage:
\begin{theorem}
Let $\S$ be a gate set and suppose \conjref{conj:densityresult} and that the polynomial hierarchy is infinite. If there exists a finite subset $\Gamma \subset \gad_1(\S)$ such that:
\begin{enumerate}[(i)]
\item for all $\delta > 0$ and all $\omega \in \Gamma$, there exists a sequence $\sigma_{\omega^{-1}}$ of gates from $\Gamma$ such that $\opnorm{\omega^{-1} - \sigma_{\omega^{-1}}} < \delta$,
\item $\words{\Gamma \cup \Gamma^{-1}}$ is a non-elementary and non-discrete subgroup of $\SL(2;\CC)$,
\end{enumerate}
then efficient classical computers cannot simulate efficient quantum computers over $\S$ to within multiplicative error $\epsilon < \sqrt{2} - 1$.
\end{theorem}

We close with a series of unrelated questions that we think are interesting:
\begin{itemize}
\item Are there any non-trivial \emph{necessary} conditions, perhaps in the form of alternative algorithms for $\IsElementary$ and $\IsDiscrete$, for a gate set $\S$ to satisfy $\GadBQP(\S) = \PostBQP$?

\item Assuming $\PostBPP \neq \PostBQP$, is there a non-universal gate set $\S$ for which $\PostBPP \subsetneq \PostBQP(\S) \subsetneq \PostBQP$?

\item For which $U \in \U(2)$ are $U$-CCCs over the $\T_4 + \Pauli$ fragment of the Clifford group efficiently classically simulable (assuming the polynomial hierarchy is infinite)? What about for the other unclassified Clifford fragments?

\item Supposing the polynomial hierarchy is infinite, does there exist a subgroup $H$ of the Clifford group that contains an entangling gate, such that for all $U \in \U(2)$, $U$-conjugated $H$ circuits are efficiently classically simulable to within multiplicative error $\epsilon < \sqrt{2} - 1$?

\end{itemize}

We hope our work inspires more research in these directions.

\subsection*{Acknowledgements} 

C.K. and M.F. thank the organizers and participants of the 2024 ``Foundations of Quantum Advantage'' conference at the Perimeter Institute, which is where work on this paper began. The authors thank Adam Bouland and Scott Aaronson for comments on an earlier draft of this paper and Luke Schaeffer for helpful discussions.

Research at Perimeter Institute is supported in part by the Government of Canada through the Department of Innovation, Science and Economic Development Canada and by the Province of Ontario through the Ministry of Colleges and Universities.

\appendix

\section{Proof of \thmref{thm:CCC_classification}}
\label{append:CCC_classification}

In this section, we reproduce the complexity classification of $U$-CCCs, which was originally done in \cite{BFK18}. Formally, this is our \thmref{thm:CCC_classification}, which we restate below for convenience.

\CCCclassification*

The proof follows from two lemmas, the second of which depends on \thmref{thm:cz+zclassification}, the complexity classification of conjugated $\CZ + Z$ circuits.

\begin{lemma}
\label{lem:CCC_lem_one}
Let $U$ and $V$ be single-qubit unitaries such that $U = e^{i\alpha} C V R_z(\lambda)$ for $C \in \words{H,S}$ and $\alpha, \lambda \in [0,2\pi)$. Then, $U$-CCCs are efficiently classically simulable to within multiplicative error $\epsilon < \sqrt{2} - 1$ iff $V$-CCCs are.
\end{lemma}

\begin{lemma}
\label{lem:CCC_lem_two}
If the polynomial hierarchy is infinite, then efficient classical computers can simulate $R_z(\phi)R_x(\theta)$-CCCs to within multiplicative error $\epsilon < \sqrt{2} - 1$ iff either $\phi \in \frac{\pi}{2}\ZZ$ and $\theta \in \frac{\pi}{2}\ZZ$, or $\theta \in \pi\ZZ$.
\end{lemma}

\begin{proof}[Proof of \thmref{thm:CCC_classification}]
Suppose $U = e^{i\alpha} CR_z(\lambda)$ for $C \in \words{H,S}$ and $\alpha, \lambda \in [0,2\pi)$. Then, by \lemref{lem:CCC_lem_one}, $U$-CCCs are efficiently classically simulable to within mulpitlicative error $\epsilon$ iff $I_1$-CCCs are. But $I_1$-CCCs are Clifford circuits, which are exactly simulable by the Gottesman-Knill theorem.

Now suppose that $U$-CCCs are efficiently classically simulable to within multiplicative error $\epsilon < \sqrt{2} - 1$, and let $U = e^{i\alpha} R_z(\phi) R_x(\theta) R_z(\lambda)$ be the Euler decomposition of $U$. By Lemmas~\ref{lem:CCC_lem_one} and \ref{lem:CCC_lem_two}, $U$-CCCs are efficiently classically simulable to within multiplicative error $\epsilon$ iff $R_z(\phi)R_x(\theta)$-CCCs are iff either $\phi \in \frac{\pi}{2}\ZZ$ and $\theta \in \frac{\pi}{2}\ZZ$, or $\theta \in \pi\ZZ$. In both cases, it is straightforward to verify that $R_z(\phi)R_x(\theta) = e^{i\alpha'} C R_z(\lambda')$ for some $C \in \words{H,S}$ and $\alpha',\lambda' \in [0,2\pi)$. Therefore, $U = e^{i\alpha + \alpha'} C R_z(\lambda + \lambda')$, which is the desired form.
\end{proof}

Thus, it remains to prove Lemmas~\ref{lem:CCC_lem_one} and \ref{lem:CCC_lem_two}. While the proof of \lemref{lem:CCC_lem_one} is straightforward, the proof of \lemref{lem:CCC_lem_two} is rather involved. Note that in this and the next section, we often write $U \sim V$ for any $U, V \in \U(2)$ that are related by $U = e^{i\alpha} V$ for some $\alpha \in [0,2\pi)$.

\begin{proof}[Proof of \lemref{lem:CCC_lem_one}]
Let $Q$ be an $n$-qubit circuit over $\S_\CCC(U)$ where, by assumption, $U = e^{i\alpha} C V R_z(\lambda)$ for some $V \in \U(2)$, $C \in \words{H, S}$, and $\alpha, \lambda \in [0,2\pi)$. Then, $Q = (U^\dagger)^{\otimes n} E U^{\otimes n}$ for some $n$-qubit Clifford circuit $E$. Since $C \in \words{H,S}$, $E' = (C^\dagger)^{\otimes n}EC^{\otimes n}$ is an $n$-qubit Clifford circuit. Therefore, $e^{i\alpha} C V R_z(\lambda)$-CCCs are efficiently classically simulable to within multiplicative error $\epsilon$ iff $e^{i\alpha} V R_z(\lambda)$-CCCs are. Finally, since for all $x,y \in \{0,1\}^n$,
$$
\left| \bra{y} (e^{-i\alpha} R_z(-\lambda) V^\dagger)^{\otimes n} E' (e^{i\alpha} V R_z(\lambda))^{\otimes n} \ket{x}\right|^2 = \left| \bra{y} (V^\dagger)^{\otimes n} E' V^{\otimes n} \ket{x}\right|^2,
$$
it holds that $e^{i\alpha} V R_z(\lambda)$-CCCs are efficiently classically simulable to within multiplicative error $\epsilon$ iff $V$-CCCs are. 
\end{proof}

\begin{proof}[Proof of \lemref{lem:CCC_lem_two}]
We break the proof up into three claims:

\begin{claim}
\label{claim:CCC_one}
If either $\phi \in \frac{\pi}{2}\ZZ$ and $\theta \in \frac{\pi}{2}\ZZ$, or $\theta \in \pi\ZZ$, then efficient classical computers can simulate $R_z(\phi)R_x(\theta)$-CCCs to within multiplicative error $\epsilon < \sqrt{2} - 1$.
\end{claim}

\begin{claim}
\label{claim:CCC_two}
Suppose the polynomial hierarchy is infinite. If $\theta \not\in \frac{\pi}{2}\ZZ$, then efficient classical computers cannot simulate $R_z(\phi)R_x(\theta)$-CCCs to within multiplicative error $\epsilon < \sqrt{2} - 1$.
\end{claim}

\begin{claim}
\label{claim:CCC_three}
Suppose the polynomial hierarchy is infinite. If $\phi \not\in \frac{\pi}{2}\ZZ$ and $\theta \in \frac{\pi}{2}\ZZ_\odd$, then efficient classical computers cannot simulate $R_z(\phi)R_x(\theta)$-CCCs to within multiplicative error $\epsilon < \sqrt{2} - 1$.
\end{claim}

Indeed, together these claims imply the lemma, because \claimref{claim:CCC_one} implies the forward direction, while the contrapositive of Claims~\ref{claim:CCC_two} and \ref{claim:CCC_three} imply the latter direction (namely, supposing that the polynomial hierarchy is infinite, if efficient classical computers can simulate $R_z(\phi)R_x(\theta)$-CCCs to within multiplicative error $\epsilon < \sqrt{2} - 1$, then either $\phi \in \frac{\pi}{2}\ZZ$ and $\theta \in \frac{\pi}{2}\ZZ_\odd \subset \frac{\pi}{2}\ZZ$, or $\theta \in \frac{\pi}{2}\ZZ_\even = \pi\ZZ$).

\begin{proof}[Proof of \claimref{claim:CCC_one}]
If $\phi \in \frac{\pi}{2} \ZZ$ and $\theta \in \frac{\pi}{2}\ZZ$, then $R_z(\phi)R_x(\theta)$ is a Clifford gate. Therefore, $R_z(\phi)R_x(\theta)$-CCCs are exactly simulable by the Gottesman-Knill theorem. If $\theta \in \pi\ZZ$, then either $\theta \in 2\pi\ZZ$ or $\theta \in \pi\ZZ_\odd$. If $\theta \in 2\pi\ZZ$, then $R_x(\theta) \sim I_1$, so $R_z(\phi)R_x(\theta)$-CCCs are exactly simulable by combining \lemref{lem:CCC_lem_one} and the Gottesman-Knill theorem. If $\theta \in \pi\ZZ_\odd$, then $R_x(\phi) \sim X$, so $R_z(\phi)R_x(\theta) \sim R_z(\phi)X = XR_z(-\phi)$. By \lemref{lem:CCC_lem_one}, $XR_z(-\phi)$-CCCs are efficiently classically simulable to within multiplicative error $\epsilon$ iff $X$-CCCs are. But $X$ is a Clifford gate, so $X$-CCCs are exactly simulable by the Gottesman-Knill theorem.
\end{proof}

\begin{proof}[Proof of \claimref{claim:CCC_two}]
For all $\phi \in [0,2\pi)$, if $\theta \not\in \frac{\pi}{2}\ZZ$, then by \thmref{thm:cz+zclassification} (whose proof is in \appendref{append:czzclassification}), efficient classical computers cannot simulate $R_z(\phi)R_x(\theta)$-conjugated $\CZ + Z$ circuits to within multiplicative error $\epsilon < \sqrt{2} - 1$. Therefore, efficient classical computers cannot simulate $R_z(\phi)R_x(\theta)$-CCCs to within the same multiplicative error.
\end{proof}

\begin{proof}[Proof of \claimref{claim:CCC_three}]
If $\theta \in \frac{\pi}{2}\ZZ_\odd$, then, up to a global phase, $R_x(\theta)$ is either $HSH$ or $HS^\dagger H$. Here we suppose $R_x(\theta) \sim HSH$, but an analogous argument works for the other case. Since $S \sim R_z(\frac{\pi}{2})$ and, more generally, $R_z(\phi) \sim R_z(\phi - \frac{\pi}{2})S$, it follows that $R_z(\phi)R_x(\theta)$-CCCs are efficiently classically simulable to within multiplicative error $\epsilon < \sqrt{2} - 1$ iff $R_z(\phi - \frac{\pi}{2})SR_x(\theta)$-CCCs are, which, by \lemref{lem:CCC_lem_one}, is true iff $R_z(\phi - \frac{\pi}{2})SR_x(\theta)S$-CCCs are. Because $R_x(\theta) \sim HSH$ and $SHSHS \sim H$, we have that $R_z(\phi - \frac{\pi}{2}) S R_x(\theta) S$-CCCs are efficiently classically simulable to within multiplicative error $\epsilon$ iff $R_z(\phi - \frac{\pi}{2}) H$-CCCs are, which, by \lemref{lem:CCC_lem_one}, is true iff $H R_z(\phi - \frac{\pi}{2}) H$-CCCs are. But $H R_z(\phi - \frac{\pi}{2}) H = R_x(\phi - \frac{\pi}{2})$, so $R_z(\phi - \frac{\pi}{2}) H$-CCCs are efficiently classically simulable to within multiplicative error $\epsilon$ iff $R_x(\phi - \frac{\pi}{2})$-CCCs are. Since $\phi \not\in \frac{\pi}{2}\ZZ,$ $\phi - \frac{\pi}{2} \not\in \frac{\pi}{2}\ZZ$. Therefore, by \thmref{thm:cz+zclassification}, efficient classical computers cannot simulate $R_x(\phi - \frac{\pi}{2})$-conjugated $\CZ + Z$ circuits to within multiplicative error $\epsilon < \sqrt{2} - 1$. Therefore, efficient classical computers cannot simulate $R_x(\phi - \frac{\pi}{2})$-CCCs to within the same multiplicative error.
\end{proof}

Altogether, Claims~\ref{claim:CCC_one} -- \ref{claim:CCC_three} complete the proof.
\end{proof}

\section{Proof of \thmref{thm:czclassification}}
\label{append:czclassification}

In this section, we classify the complexity of $U$-conjugated $\CZ$ circuits. Formally, this is our \thmref{thm:czclassification}, which we restate below for convenience.

\CZclassification*

Our method of proof is to show that for all $U \in \U(2)$, $U$-conjugated $\CZ$ circuits are efficiently classically simuilable iff $U$-conjugated $\CZ + Z$ circuits are, thereby reducing the problem to finding the complexity classification of $U$-conjugated $\CZ + Z$ circuits, which we do in \appendref{append:czzclassification}. To show this reduction, we rely on \lemref{lem:CZlemma} and the following lemma.

\begin{lemma}
\label{lem:czz_lem_one}
Let $U$ and $V$ be single-qubit unitaries such that $U = e^{i\alpha} R_z(\phi) V R_z(\lambda)$ for $\alpha, \phi, \lambda \in [0,2\pi)$. Then, $U$-conjugated $\CZ + Z$ circuits are efficiently classically simulable to within multiplicative error $\epsilon < \sqrt{2} - 1$ iff $V$-conjugated $\CZ + Z$ circuits are.
\end{lemma}
\begin{proof}
Let $Q$ be an $n$-qubit circuit over $\S_{\CZ + Z}(U)$ where, by assumption, $U = e^{i\alpha} R_z(\phi) V R_z(\lambda)$ for some $V \in \U(2)$. Then, $Q = (U^\dagger)^{\otimes n} D U^{\otimes n}$ for some $n$-qubit circuit $D$ that is diagonal in the computational basis. Consequently, $R_z(-\phi)^{\otimes n} D R_z(\phi)^{\otimes n} = D$, so $e^{i\alpha} R_z(\phi) V R_z(\lambda)$-conjugated $\CZ + Z$ circuits are efficiently classically simulable to within multiplicative error $\epsilon$ iff $e^{i\alpha} V R_z(\lambda)$-conjugated $\CZ + Z$ circuits are. Finally, since for all $x,y \in \{0,1\}^n$,
$$
\left| \bra{y} (e^{-i\alpha} R_z(-\lambda) V^\dagger)^{\otimes n} D (e^{i\alpha} V R_z(\lambda))^{\otimes n} \ket{x}\right|^2 = \left| \bra{y} (V^\dagger)^{\otimes n} D V^{\otimes n} \ket{x}\right|^2,
$$
it holds that $e^{i\alpha} V R_z(\lambda)$-conjugated $\CZ + Z$ circuits are efficiently classically simulable to within multiplicative error $\epsilon$ iff $V$-conjugated $\CZ + Z$ circuits are.
\end{proof}

We now prove that the complexity classifications of conjugated $\CZ$ circuits and conjugated $\CZ + Z$ circuits are the same. Consequently, \thmref{thm:czclassification} follows from \thmref{thm:cz+zclassification}, which is proved in \appendref{append:czzclassification}.

\begin{lemma}
\label{lem:czequivalence}
If the polynomial hierarchy is infinite, then for all $U \in \U(2)$, efficient classical computers can simulate $U$-conjugated $\CZ + Z$ circuits to within multiplicative error $\epsilon < \sqrt{2} - 1$ iff efficient classical computers can simulate $U$-conjugated $\CZ$ circuits to within the same multiplicative error.
\end{lemma}
\begin{proof}
The forward direction is trivial, because every $U$-conjugated $\CZ$ circuit is a $U$-conjugated $\CZ + Z$ circuit. Now suppose efficient classical computers can simulate $U$-conjugated $\CZ$ circuits to within multiplicative error $\epsilon < \sqrt{2} - 1$, and let $U = e^{i\alpha} R_z(\phi) R_x(\theta) R_z(\lambda)$ be the Euler decomposition of $U$. Then, by \lemref{lem:czz_lem_one}, efficient classical computers can simulate $U$-conjugated $\CZ + Z$ circuits to within multiplicative error $\epsilon$ iff they can simulate $R_x(\theta)$-conjugated $\CZ + Z$ circuits. If $\theta \in \frac{\pi}{2}\ZZ$, then $R_x(\theta) \in \words{H,S}$, so every $R_x(\theta)$-conjugated $\CZ + Z$ circuit and every $R_x(\theta)$-conjugated $\CZ$ circuit is a Clifford circuit, so both types are exactly simulable by the Gottesman-Knill theorem. If $\theta \not\in\frac{\pi}{2}\ZZ$, then by \lemref{lem:CZlemma}, efficient classical computers can simulate $R_x(\theta)$-conjugated $\CZ + Z$ circuits to within multiplicative error $\epsilon$ iff they can simulate $R_x(\theta)$-conjugated $\CZ$ circuits to within the same multiplicative error.
\end{proof}

\section{Proof of \thmref{thm:cz+zclassification}}
\label{append:czzclassification}

In this section, we classify the complexity of $U$-conjugated $\CZ + Z$ circuits. Formally, this is our \thmref{thm:cz+zclassification}, which we restate below for convenience.

\CZZclassification*

The proof follows from \lemref{lem:czz_lem_one} and the following lemma.
\begin{lemma}
\label{lem:czz_lem_two}
If the polynomial hierarchy is infinite, then efficient classical computers can simulate $R_x(\theta)$-conjugated $\CZ + Z$ circuits to within multiplicative error $\epsilon < \sqrt{2} - 1$ iff $\theta \in \frac{\pi}{2}\ZZ$.
\end{lemma}

\begin{proof}[Proof of \thmref{thm:cz+zclassification}]
Suppose $U = e^{i\alpha} R_z(\phi) C R_z(\lambda)$ for $C \in \words{H,S}$ and $\alpha, \phi, \lambda \in [0,2\pi)$. Then, by \lemref{lem:czz_lem_one}, $U$-conjugated $\CZ + Z$ circuits are efficiently classically simulable to within multiplicative error $\epsilon$ iff $C$-conjugated $\CZ + Z$ circuits are. But $C$ is a Clifford operation, so $C$-conjugated $\CZ + Z$ circuits are exactly simulable by the Gottesman-Knill theorem.

Now suppose that $U$-conjugated $\CZ + Z$ circuits are efficiently classically simulable to within multiplicative error $\epsilon < \sqrt{2} - 1$, and let $U = e^{i\alpha} R_z(\phi) R_x(\theta) R_z(\lambda)$ be the Euler decomposition of $U$. By Lemmas~\ref{lem:czz_lem_one} and \ref{lem:czz_lem_two}, $U$-conjugated $\CZ + Z$ circuits are efficiently classically simulable to within multiplicative error $\epsilon$ iff $R_x(\theta)$-conjugated $\CZ + Z$ circuits are iff $\theta \in \frac{\pi}{2}\ZZ$. Therefore, $R_x(\theta) \in \words{H,S}$, so $U = e^{i\alpha} R_z(\phi) C R_z(\lambda)$ for $C = R_x(\theta) \in \words{H,S}$, as desired.
\end{proof}

Thus, it remains to prove \lemref{lem:czz_lem_two}. While the proof uses our criterion, it is rather involved because it is arithmetically tedious. Because of this, we made extensive use of Wolfram Mathematica 14.1. Our notebook is available online \cite{KFG24}.

\begin{proof}[Proof of \lemref{lem:czz_lem_two}]
Suppose $\theta \in \frac{\pi}{2} \ZZ$. Then, $R_x(\theta) \in \words{H,S}$, so the $R_x(\theta)$-conjugated $\CZ + Z$ circuit is a Clifford circuit and is therefore exactly simulable by the Gottesman-Knill theorem.

Now suppose $\theta \not \in \frac{\pi}{2}\ZZ$. We will show that efficient classical computers cannot simulate $R_x(\theta)$-conjugated $\CZ + Z$ circuits to within multiplicative error $\epsilon < \sqrt{2} - 1$ unless the polynomial hierarchy collapses. To this end, it suffices to exhibit a finite number of post-selection gadgets over $\S_{\CZ + Z}(R_x(\theta))$ such that their normalized actions generate a non-elementary, non-discrete, and strictly loxodromic subgroup of $\SL(2;\CC)$. The gadgets and their inverses are presented in \tabref{tab:czz_CCC_gadgets}. The determinants of the actions of these gadgets are:
\begin{align*}
\det \AA(\c_0) &= \det \AA(\c_0^{-1}) = \frac{\sin^4(\theta)}{4}\\
\det \AA(\c_1) &= \det \AA(\c_1^{-1}) = \frac{1}{32}\left( 5 - 28 \cos(\theta) - 4 \cos(2\theta) - 4\cos(3\theta) - \cos(4\theta) \right)\\
\det \AA(\c_2) &= \det \AA(\c_2^{-1}) = \frac{1}{32}\left( 5 + 28 \cos(\theta) - 4 \cos(2\theta) + 4\cos(3\theta) - \cos(4\theta) \right)\\
\det \AA(\c_3) &= \det \AA(\c_3^{-1}) = \cos(\theta)\\
\det \AA(\c_4) &= \det \AA(\c_4^{-1}) = -\cos(\theta).
\end{align*}

Evidently, if $\theta \not\in \frac{\pi}{2} \ZZ$, then $\det \AA(\c_i) \neq 0$ for $i \in \{0,3,4 \}$, so the \emph{normalized} actions of $\c_0$, $\c_3$, and $\c_4$ exist for such $\theta$. They are:
\begin{align*}
C_0 \defeq \tilde{\AA}(\c_0) &= \frac{1}{4}
\begin{pmatrix}
 (-\cos (\theta )+2 \cos (2 \theta )+\cos (3 \theta )+6) \csc ^2(\theta ) & \frac{4 i \sin (\theta ) \cos (\theta )}{\cos (\theta )-1} \\
 2 i \sin (\theta ) \cos (\theta ) \csc ^2\left(\frac{\theta }{2}\right) & 4 (\cos (\theta )+1)
\end{pmatrix}
\\
C_3 \defeq \tilde{\AA}(\c_3) &= \frac{1}{\sqrt{\cos(\theta)}}
\begin{pmatrix}
 \frac{1}{4} (\cos (2 \theta )+3) & i \sin ^2\left(\frac{\theta }{2}\right) \sin (\theta ) \\
 -i \sin ^2\left(\frac{\theta }{2}\right) \sin (\theta ) & \frac{1}{2} \left(\sin ^2(\theta )+2 \cos (\theta )\right)
\end{pmatrix}
\\
C_4 \defeq \tilde{\AA}(\c_4) &= \frac{\sqrt{-\cos(\theta)}}{2}
\begin{pmatrix}
 -\frac{1}{2} (\cos (2 \theta )+3) \sec (\theta ) & i (\sin (\theta )+\tan (\theta )) \\
 -i (\sin (\theta )+\tan (\theta )) & \frac{1}{2} (4 \cos (\theta )+\cos (2 \theta )-1) \sec (\theta )
\end{pmatrix}.
\end{align*}
However, since $\det \AA(\c_1) = 0$ iff $\theta \in 2\pi\ZZ \pm 2\tan^{-1}\left( \sqrt{\sqrt{2} - 1} \right)$ and $\det \AA(\c_2) = 0$ iff $\theta \in 2\pi\ZZ \pm 2\tan^{-1}\left( \sqrt{\sqrt{2} + 1} \right)$, we will proceed on a case-by-case basis. To do this, we partition $[0,2\pi)$ into the two disjoint intervals $A \defeq (\frac{\pi}{2}, \frac{3\pi}{2})$ and $ B \defeq [0,\frac{\pi}{2}] \cup [\frac{3\pi}{2}, 2\pi)$ and prove the lemma in each interval.

If $\theta \in A$, then $\det \AA(\c_1) \neq 0$, so the normalized of action of $\c_1$ exists for such $\theta$. It is:
$$
C_1 \defeq \tilde{\AA}(\c_1) = \frac{\AA(\c_1)}{\sqrt{\det \AA(\c_1)}},
$$
where
$$
\AA(\c_1) = 
\frac{1}{8}
\begin{pmatrix}
 -\cos (\theta )+2 \cos (2 \theta )+\cos (3 \theta )+6 & -i (\sin (\theta )+2 \sin (2 \theta )+\sin (3 \theta )) \\
 i (\sin (\theta )+2 \sin (2 \theta )+\sin (3 \theta )) & -7 \cos (\theta )-2 \cos (2 \theta )-\cos (3 \theta )+2
\end{pmatrix}.
$$
Similarly, if $\theta \in B$, then $\det \AA(\c_2) \neq 0$, so the normalized action of $\c_2$ exists for such $\theta$. It is:
$$
C_2 \defeq 
\tilde{\AA}(\c_2) = \frac{\AA(\c_2)}{\sqrt{\det \AA(\c_2)}},
$$
where
$$
\AA(\c_2) = 
\frac{1}{8}
\begin{pmatrix}
 7 \cos (\theta )-2 \cos (2 \theta )+\cos (3 \theta )+2 & -i (\sin (\theta )-2 \sin (2 \theta )+\sin (3 \theta )) \\
 -8 i \sin ^2\left(\frac{\theta }{2}\right) \sin (\theta ) \cos (\theta ) & \cos (\theta )+2 \cos (2 \theta )-\cos (3 \theta )+6
\end{pmatrix}.
$$

We underscore that for all values of $\theta$ where the normalized action of $\c_i$ is defined, the gadget $\c_i^{-1}$ given in \tabref{tab:czz_CCC_gadgets} is also defined and an inverse of $\c_i$:
\begin{align*}
\tilde{\AA}(\c_0^{-1}) &= \tilde{\AA}(\c_0)^{-1} = C_0^{-1}\\
\tilde{\AA}(\c_1^{-1}) &= \tilde{\AA}(\c_1)^{-1} = C_1^{-1}\\
-\tilde{\AA}(\c_2^{-1}) &= \tilde{\AA}(\c_2)^{-1} = C_2^{-1}\\
\tilde{\AA}(\c_3^{-1}) &= \tilde{\AA}(\c_3)^{-1} = C_3^{-1}\\
\tilde{\AA}(\c_4^{-1}) &= \tilde{\AA}(\c_4)^{-1} = C_4^{-1}
\end{align*}

Now let
$$
\Gamma^A_{\CZ + Z}(R_x(\theta)) \defeq \big\{C_0, C_0^{-1}, C_1, C_1^{-1}, C_3, C_3^{-1}, C_4, C_4^{-1}\big\}
$$ 
and
$$
\Gamma^B_{\CZ + Z}(R_x(\theta)) \defeq \big\{C_0, C_0^{-1}, C_2, -C_2^{-1}, C_3, C_3^{-1}, C_4, C_4^{-1}\big\}.
$$ 

Evidently, if $\theta \in A \backslash \frac{\pi}{2}\ZZ$, then $\Gamma^A_{\CZ + Z}(R_x(\theta))$ is closed under inverses, so $\words{\Gamma^A_{\CZ + Z}(R_x(\theta))}$ is a subgroup of $\SL(2;\CC)$. Likewise, if $\theta \in B \backslash \frac{\pi}{2}\ZZ$, then $\words{\Gamma^B_{\CZ + Z}(R_x(\theta))}$ is a subgroup of $\SL(2;\CC)$. In fact, for the appropriate choice of $\theta$, both $\words{\Gamma^A_{\CZ + Z}(R_x(\theta))}$ and $\words{\Gamma^B_{\CZ + Z}(R_x(\theta))}$ are non-elementary, non-discrete, and strictly loxodromic subgroups of $\SL(2;\CC)$.

\afterpage{\clearpage}
\begin{table}[htpb]
    \centering
    \begingroup
    \renewcommand*{\arraystretch}{1.5}
     \begin{tabular}{|c|c||c|c|}
        \hline
        $\c_i$ & Gadget in $\S_{\CZ + Z}(U)$ & $\c_i^{-1}$ & Inverse gadget in $\S_{\CZ + Z}(U)$ 
        \\
        \hline\hline
$\c_0$ &
    \begin{quantikz}
    \\
        \lstick{\ket{0}} & \gate{U} & \ctrl{1} & \gate{Z} & \gate{U^\dagger} & \rstick{} \\
        \lstick{\ket{0}} & \gate{U} & \ctrl{0} & \ctrl{1} & \gate{U^\dagger} & \rstick{\bra{0}} \\
        \lstick{} & \gate{U} & \gate{Z} & \ctrl{0} & \gate{U^\dagger} & \rstick{\bra{0}} 
    \\
    \end{quantikz}  
        & 
$\c_0^{-1}$ &        
    \begin{quantikz}
    \\
        \lstick{} & \gate{U} & \ctrl{1} & \gate{Z} & \gate{U^\dagger} & \rstick{\bra{1}} \\
        \lstick{\ket{0}} & \gate{U} & \ctrl{0} & \ctrl{1} & \gate{U^\dagger} & \rstick{\bra{0}} \\
        \lstick{\ket{1}} & \gate{U} & \gate{Z} & \ctrl{0} & \gate{U^\dagger} & \rstick{} 
    \\
    \end{quantikz} \\ \hline
$\c_1$ & 
    \begin{quantikz}
    \\
        \lstick{} & \gate{U} & \ctrl{1} & \gate{Z} & \gate{U^\dagger} & \rstick{} \\
        \lstick{\ket{0}} & \gate{U} & \ctrl{0} & \ctrl{1} & \gate{U^\dagger} & \rstick{\bra{0}} \\
        \lstick{\ket{0}} & \gate{U} & \gate{Z} & \ctrl{0} & \gate{U^\dagger} & \rstick{\bra{0}} 
    \\
    \end{quantikz}  
    & 
$\c_1^{-1}$ & 
    \begin{quantikz}
    \\
        \lstick{} & \gate{U} & \ctrl{1} &  & \gate{U^\dagger} & \rstick{} \\
        \lstick{\ket{1}} & \gate{U} & \ctrl{0} & \ctrl{1} & \gate{U^\dagger} & \rstick{\bra{1}} \\
        \lstick{\ket{1}} & \gate{U} &  & \ctrl{0} & \gate{U^\dagger} & \rstick{\bra{1}} 
    \\
    \end{quantikz} \\ \hline
$\c_2$ &
    \begin{quantikz}
    \\
        \lstick{} & \gate{U} & & \ctrl{1} & & \gate{U^\dagger} & \rstick{} \\
        \lstick{\ket{0}} & \gate{U} & \gate{Z} & \ctrl{0} & \ctrl{1} & \gate{U^\dagger} & \rstick{\bra{0}} \\
        \lstick{\ket{0}} & \gate{U} & & & \ctrl{0} & \gate{U^\dagger} & \rstick{\bra{0}}
    \\
    \end{quantikz}  
    & 
$\c_2^{-1}$ &
    \begin{quantikz}
    \\
        \lstick{} & \gate{U} & \gate{Z} & \ctrl{1} & & \gate{U^\dagger} & \rstick{} \\
        \lstick{\ket{1}} & \gate{U} & \gate{Z} & \ctrl{0} & \ctrl{1} & \gate{U^\dagger} & \rstick{\bra{1}} \\
        \lstick{\ket{1}} & \gate{U} & \gate{Z} & & \ctrl{0} & \gate{U^\dagger} & \rstick{\bra{1}}
    \\
    \end{quantikz} \\ \hline

$\c_3$ &
    \begin{quantikz}
    \\
        \lstick{} & \gate{U} & & \ctrl{1} & \gate{U^\dagger} & \rstick{} \\
        \lstick{\ket{0}} & \gate{U} & \gate{Z} & \ctrl{0} & \gate{U^\dagger} & \rstick{\bra{0}}
    \\
    \end{quantikz}  
    & 
$\c_3^{-1}$ &
    \begin{quantikz}
    \\
        \lstick{} & \gate{U} & \gate{Z} & \ctrl{1} & \gate{U^\dagger} & \rstick{} \\
        \lstick{\ket{1}} & \gate{U} &  & \ctrl{0} & \gate{U^\dagger} & \rstick{\bra{1}}
    \\
    \end{quantikz} \\ \hline

    $\c_4$ &
    \begin{quantikz}
    \\
        \lstick{} & \gate{U} & \gate{Z} & \ctrl{1} & \gate{U^\dagger} & \rstick{} \\
        \lstick{\ket{0}} & \gate{U} & & \ctrl{0} & \gate{U^\dagger} & \rstick{\bra{0}}
    \\
    \end{quantikz}  
    & 
$\c_4^{-1}$ &
    \begin{quantikz}
    \\
        \lstick{} & \gate{U} & & \ctrl{1} & \gate{U^\dagger} & \rstick{} \\
        \lstick{\ket{1}} & \gate{U} & \gate{Z} & \ctrl{0} & \gate{U^\dagger} & \rstick{\bra{1}}
    \\
    \end{quantikz} \\ \hline
\end{tabular}
\endgroup
\caption{Post-selection gadgets for $U$-conjugated $\CZ + Z$ Clifford circuits with $U = R_x(\theta)$ and $\theta \not \in \frac{\pi}{2}\ZZ$.}
\label{tab:czz_CCC_gadgets}
\end{table}

\begin{claim}
\label{claim:czz_CCC_1}
If $\theta \in A \backslash \frac{\pi}{2} \ZZ$, then $\words{\Gamma^A_{\CZ + Z}(R_x(\theta))}$ is a non-elementary subgroup of $\SL(2;\CC)$.
\end{claim}

\begin{proof}
We find
\begin{align*}
\beta(C_0) &= 4(\csc^4(\theta) - 1)\\
\beta(C_1) &= \frac{4}{\sec^2(\theta) \tan^4 \left( \frac{\theta}{2} \right)}\\
\gamma(C_0, C_1) &= \frac{32\cos^4(\theta)\cot^2\left( \frac{\theta}{2} \right)}{-5 + 28\cos(\theta) + 4 \cos(2\theta) + 4\cos(3\theta) + \cos(4\theta)}.
\end{align*}
It is an elementary exercise to verify that if $\theta \in A \backslash \frac{\pi}{2}\ZZ$, then $\IsElementary(\Gamma^A_{\CZ + Z}(R_x(\theta))) = \NO$.
\end{proof}

\begin{claim}
\label{claim:czz_CCC_12}
If $\theta \in B \backslash \frac{\pi}{2} \ZZ$, then $\words{\Gamma^B_{\CZ + Z}(R_x(\theta))}$ is a non-elementary subgroup of $\SL(2;\CC)$.
\end{claim}
\begin{proof}
We find
\begin{align*}
\beta(C_0) &= 4(\csc^4(\theta) - 1)\\ 
\beta(C_2) &= \frac{128\cos^2(\theta)\sin^4 \left( \frac{\theta}{2} \right)}{-5 - 28\cos(\theta) + 4 \cos(2\theta) - 4\cos(3\theta) + \cos(4\theta)}\\
\gamma(C_0, C_2) &= \frac{32\cos^4(\theta)\cot^2\left( \frac{\theta}{2} \right)}{-5 - 28\cos(\theta) + 4 \cos(2\theta) - 4\cos(3\theta) + \cos(4\theta)}.
\end{align*}
It is an elementary exercise to verify that if $\theta \in B \backslash \frac{\pi}{2}\ZZ$, then $\IsElementary(\Gamma^B_{\CZ + Z}(R_x(\theta))) = \NO$.
\end{proof}

\begin{claim}
\label{claim:czz_CCC_2}
If $\theta \in A \backslash \frac{\pi}{2} \ZZ$, then $\words{\Gamma^A_{\CZ + Z}(R_x(\theta))}$ is a non-discrete subgroup of $\SL(2;\CC)$.
\end{claim}
\begin{proof}
From the proof of \claimref{claim:czz_CCC_1}, we know that if $\theta \in A \backslash \frac{\pi}{2}\ZZ$, then $C_0$ and $C_1$ (together with their inverses) generate a non-elementary subgroup of $\SL(2;\CC)$. Therefore, we may use J{\o}rgensen's inequality (line 4 in \algref{alg:discrete}) to test discreteness. We find that if $\theta \in A \backslash \frac{\pi}{2}\ZZ$, then indeed 
$$
\left|\tr^2(C_1) - 4\right| + \left|\tr(C_1 C_0 C_1^{-1} C_0^{-1}) - 2 \right| = \frac{32 \left(\cos ^4(\theta ) \left| \cot \left(\frac{\theta }{2}\right)\right| ^2+4 \cos ^2(\theta ) \cos ^4\left(\frac{\theta }{2}\right)\right)}{| 28 \cos (\theta )+4 \cos (2 \theta )+4 \cos (3 \theta )+\cos (4 \theta )-5| } < 1.
$$
Therefore, if $\theta \in A \backslash \frac{\pi}{2}\ZZ$, then $\IsDiscrete(\Gamma^A_{\CZ + Z}(R_x(\theta))) = \NO$.
\end{proof}

\begin{claim}
\label{claim:czz_CCC_21}
If $\theta \in B \backslash \frac{\pi}{2} \ZZ$, then $\words{\Gamma^B_{\CZ + Z}(R_x(\theta))}$ is a non-discrete subgroup of $\SL(2;\CC)$.
\end{claim}
\begin{proof}
From the proof of \claimref{claim:czz_CCC_12}, we know that if $\theta \in B \backslash \frac{\pi}{2}\ZZ$, then $C_0$ and $C_2$ (together with their inverses) generate a non-elementary subgroup of $\SL(2;\CC)$. Therefore, we may use J{\o}rgensen's inequality (line 4 in \algref{alg:discrete}) to test discreteness. We find that if $\theta \in B \backslash \frac{\pi}{2}\ZZ$, then indeed 
$$
\left|\tr^2(C_2) - 4\right| + \left|\tr(C_2 C_0 C_2^{-1} C_0^{-1}) - 2 \right| = \frac{32 \left(\cos ^4(\theta ) \left| \tan \left(\frac{\theta }{2}\right)\right| ^2+4 \sin ^4\left(\frac{\theta }{2}\right) \cos ^2(\theta )\right)}{| -28 \cos (\theta )+4 \cos (2 \theta )-4 \cos (3 \theta )+\cos (4 \theta )-5| } < 1.
$$
Therefore, if $\theta \in B \backslash \frac{\pi}{2}\ZZ$, then $\IsDiscrete(\Gamma^B_{\CZ + Z}(R_x(\theta))) = \NO$.
\end{proof}

\begin{claim}
\label{claim:czz_CCC_3}
If $\theta \in A \backslash \frac{\pi}{2}\ZZ$, then $\words{\Gamma^A_{\CZ + Z}(R_x(\theta))}$ is a strictly loxodromic subgroup of $\SL(2; \CC)$.
\end{claim}
\begin{proof}
Since 
$$
\tr(C_3) = \frac{1 + \cos(\theta)}{\sqrt{\cos(\theta)}} \quad \text{and} \quad
\tr(C_4) = (\cos(\theta) - 1)\sqrt{-\cos(\theta)}\sec(\theta),
$$
if $\theta \in A \backslash \frac{\pi}{2}\ZZ$, then either $\sqrt{\cos(\theta)} \in \CC \backslash \RR$ or $\sqrt{-\cos(\theta)} \in \CC \backslash \RR$. Consequently, if $\theta \in A \backslash \frac{\pi}{2}\ZZ$, then $\IsLoxodromic(\Gamma^A_{\CZ + Z}(R_x(\theta))) = \YES$.
\end{proof}

\begin{claim}
\label{claim:czz_CCC_31}
If $\theta \in B \backslash \frac{\pi}{2}\ZZ$, then $\words{\Gamma^B_{\CZ + Z}(R_x(\theta))}$ is a strictly loxodromic subgroup of $\SL(2; \CC)$.
\end{claim}
\begin{proof}
The proof is identical to the proof of \claimref{claim:czz_CCC_3}.
\end{proof}

Altogether, Claims~\ref{claim:czz_CCC_1} -- \ref{claim:czz_CCC_31} and our criterion \thmref{thm:maintwo} imply that if $\theta \not \in \frac{\pi}{2}\ZZ$, then no efficient classical computer can simulate $R_x(\theta)$-conjugated $\CZ + Z$ circuits to within multiplicative error $\epsilon < \sqrt{2} - 1$.
\end{proof}

\section{Proof of \thmref{thm:cz+sclassification}}
\label{append:czsclassification}

In this section, we classify the complexity of $U$-conjugated $\CZ + S$ circuits (also known as commuting conjugated Clifford circuits). Formally, this is our \thmref{thm:cz+sclassification}, which we restate below for convenience.

\CZSclassification*

Our method of proof is to show that for all $U \in \U(2)$, $U$-conjugated $\CZ + S$ circuits are efficiently classically simuilable iff $U$-conjugated $\CZ + Z$ circuits are, thereby showing that the complexity classification of $U$-conjugated $\CZ + S$ circuits is the same as the complexity classification of $U$-conjugated $\CZ + Z$ circuits, which we proved in \appendref{append:czzclassification}. To do this, we require the following lemma.

\begin{lemma}
\label{lem:czs_lem_one}
Let $U$ and $V$ be single-qubit unitaries such that $U = e^{i\alpha} R_z(\phi) V R_z(\lambda)$ for $\alpha, \phi, \lambda \in [0,2\pi)$. Then, $U$-conjugated $\CZ + S$ circuits are efficiently classically simulable to within multiplicative error $\epsilon < \sqrt{2} - 1$ iff $V$-conjugated $\CZ + S$ circuits are.
\end{lemma}
\begin{proof}
The proof is identical to the proof of \lemref{lem:czz_lem_one}.
\end{proof}

\begin{lemma}
\label{lem:czs_lem_two}
If the polynomial hierarchy is infinite, then for all $U \in \U(2)$, efficient classical computers can simulate $U$-conjugated $\CZ + S$ circuits to within multiplicative error $\epsilon < \sqrt{2} - 1$ iff they can simulate $U$-conjugated $\CZ + Z$ circuits to within the same multiplicative error.
\end{lemma}
\begin{proof}
The forward direction is trivial, because every $U$-conjugated $\CZ + Z$ circuit is a $U$-conjugated $\CZ + S$ circuit. For the other direction, suppose efficient classical computers can simulate $U$-conjugated $\CZ + Z$ circuits to within multiplicative error $\epsilon < \sqrt{2} - 1$ and let $U = e^{i\alpha} R_z(\phi) R_x(\theta) R_z(\lambda)$ be the Euler decomposition of $U$. By \lemref{lem:czz_lem_one}, $U$-conjugated $\CZ + Z$ circuits are efficiently classically simulable to within multiplicative error $\epsilon$ iff $R_x(\theta)$-conjugated $\CZ + Z$ circuits are. But, by \lemref{lem:czz_lem_two}, efficient classical computers can simulate $R_x(\theta)$-conjugated $\CZ + Z$ circuits to within multiplicative error $\epsilon < \sqrt{2} - 1$ iff $\theta \in \frac{\pi}{2}\ZZ$. Therefore, $\theta \in \frac{\pi}{2}\ZZ$, so $R_x(\theta) \in \words{H,S}$. Consequently, every $R_x(\theta)$-conjugated $\CZ + S$ circuit is a Clifford circuit and hence is efficiently classically simulable. Thus, by \lemref{lem:czs_lem_one}, $U$-conjugated $\CZ + S$ circuits are also efficiently classically simulable.
\end{proof}


\bibliographystyle{alpha}
\bibliography{references}

\end{document}